\documentclass[onefignum,onetabnum]{siamonline250211}

\usepackage{lipsum}
\usepackage{amsfonts}
\usepackage{graphicx}
\usepackage{epstopdf}
\usepackage{algorithmic}
\ifpdf
  \DeclareGraphicsExtensions{.eps,.pdf,.png,.jpg}
\else
  \DeclareGraphicsExtensions{.eps}
\fi

\usepackage{enumitem}
\setlist[enumerate]{leftmargin=.5in}
\setlist[itemize]{leftmargin=.5in}

\newsiamremark{remark}{Remark}
\newsiamremark{hypothesis}{Hypothesis}
\crefname{hypothesis}{Hypothesis}{Hypotheses}
\newsiamthm{claim}{Claim}
\newsiamremark{fact}{Fact}
\crefname{fact}{Fact}{Facts}

\headers{Dynamic solutions of next generation neural field models with delays}{O. E. Omel{'}chenko and C. R. Laing}

\title{Dynamic solutions of next generation neural field models with delays
\thanks{Submitted to the editors DATE.
\funding{The work of O.E.O. was supported by the Deutsche Forschungsgemeinschaft (DFG)
under Grant No. OM 99/2-3.}}}

\author{Oleh E. Omel'chenko\thanks{Institute of Physics and Astronomy, University of Potsdam, Potsdam, Germany 
  (\email{omelchenko@uni-potsdam.de}).}
\and Carlo R. Laing\thanks{School of Mathematical and Computational Sciences, 
Massey University, Auckland, New Zealand 
  (\email{c.r.laing@massey.ac.nz}).}
}

\usepackage{amsopn}

\ifpdf
\hypersetup{
  pdftitle={Dynamic solutions of next generation neural field models with delays},
  pdfauthor={O. E. Omel'chenko and C. R. Laing}
}
\fi

\newcommand{\fr}[2]{\frac{\displaystyle #1}{\displaystyle #2}}
\newcommand{\df}[2]{\frac{\displaystyle d#1}{\displaystyle d#2}}

\newcommand{\pf}[2]{\frac{\displaystyle \partial #1}{\displaystyle \partial #2}}

\newcommand{\Real}{\:\mathrm{Re}\:}
\newcommand{\Imag}{\:\mathrm{Im}\:}

\newcommand{\be}{\begin{equation}}
\newcommand{\ee}{\end{equation}}
\newcommand{\bp}{\begin{pmatrix}}
\newcommand{\ep}{\end{pmatrix}}

\newcommand{\la}{\langle}
\newcommand{\ra}{\rangle}

\begin{document}

\maketitle

% REQUIRED
\begin{abstract}
We study networks of theta neurons arranged on a ring with delayed interactions. In the continuum limit the systems
are described by next generation neural field models with delays. We consider distributed delays with both
finite and infinite support, and conduction delays. The stability of spatially uniform and localized bump states
is determined, and we find that they undergo Hopf bifurcations as parameters related to the delays are varied.
These bifurcations create traveling waves and ``breathing'' bump solutions. These dynamic solutions satisfy
self-consistency equations and we show how to efficiently solve these equations. Following
traveling waves and periodic solutions as parameters are varied provides a global picture 
of the influence of different delays on pattern formation processes
in spatially extended networks of theta neurons.
 
\end{abstract}

% REQUIRED

% REQUIRED

\section{Introduction}

Time delays are ubiquitous in neural systems. Action potentials travel at finite speeds along axons and
dendrites, and synaptic processing is not instantaneous~\cite{cam07,gab09,coolai09}.
Many authors have studied the effects of delays in neural 
models~\cite{golerm00,visnic17,devmon18,rahbly15,aldcam24}, 
which often take the form  of delay differential
equations (DDEs). Often, the only analytical progress that can be made in the study of DDEs
is the determination of the stability of a steady state, although periodic solutions
can be constructed in special cases~\cite{coolai09,laikra25}.
The accurate numerical solution of DDEs requires dedicated software such as the 
Matlab solver {\tt dde23}~\cite{shatom01}, and numerical bifurcation analysis of DDEs can be performed
using software such as DDE-BIFTOOL~\cite{DDEBIF}.

One well-studied type of neural model is the neural field model, describing activity at a spatial
scale much greater than that of a single neuron~\cite{coo05,coogra14,ama77,bre12,wilcow73}. Such models
typically take the form of an integro-differential equation, involving an integral over space to model
long range connections between neurons. Many authors have studied the effects of delays
in these classical neural field
models~\cite{roxbru05,atahut06,atahut05,vencoo07,fayfau10,vel13,roxmon11,qubjir09}.

However, the neural field models just mentioned are essentially phenomenological, not rigorously
derived from networks of spiking neurons. In contrast, the last decade has seen the introduction and study of {\em
next generation} neural field models, derived from networks of theta 
neurons~\cite{lai14a,lai15,laiome23,coobyr19,byravi19,esnrox17,schavi20,coo23}. The derivation relies on
the Ott/Antonsen ansatz~\cite{ottant08,ottant09}, a result showing the existence of a low-dimensional
attracting manifold on which the dynamics of infinite heterogeneous networks of phase oscillators
can be described by differential equations governing a small number of macroscopic variables.
One of the first applications of the Ott/Antonsen ansatz was to the study of chimera 
states in networks of sinusoidally coupled Kuramoto phase oscillators~\cite{abrmir08,lai09B}, but theta
neurons have the same mathematical form as Kuramoto oscillators and the ansatz is equally
applicable to infinite networks of them~\cite{lukbar13,lai14a}. It is of interest to compare
the effects of delays in classical neural field models with the effects in next generation models.

Next generation neural field models take the form of a Riccati equation, i.e.,~the time derivative
of a complex-valued variable is given by a quadratic function of that variable, with
possibly time-dependent coefficients~\cite{laiome23}.
Thus periodic solutions of these models are periodic solutions of periodically-forced  Riccati equations.
If the periodic forcing is prescribed, the periodic solution of interest can be found in a computationally
efficient manner using a small number of numerical integrations of Riccati equations~\cite{ome23,ome22,laiome23}.
In previous work we used this technique to study periodic solutions in 
next generation neural field models~\cite{laiome23}. One of the models studied there had a 
single discrete delay
and we followed periodic solution of the model as the delay was varied.

In this paper we greatly extend the results in~\cite{laiome23} to models in which there is a continuum 
of delays, either distributed delays or conduction delays. For each of three models, we semi-analytically
determine the stability of spatially uniform and bump states. We show how to derive equations, 
the solutions of which self-consistently describe
traveling waves and breathing bump solutions, and how to efficiently solve these equations. Following 
solutions of these equations as parameters are varied we find that traveling waves and breathing bumps are created in Hopf bifurcations
from spatially uniform and bump states, respectively. In many cases we determine the stability of solution
branches by direct numerical simulation of the corresponding neural field model.
Note that the use of self-consistency arguments to study the existence of solutions in infinite networks
of phase oscillators goes back to at least Kuramoto~\cite{kuramotobook84}.

Our basic model is a system
of $N$ theta neurons described by the equations
\be
   \frac{d\theta_j}{dt}=1-\cos{\theta_j}+(1+\cos{\theta_j})(\eta_j+\kappa I_j(t)),
   \qquad j=1,2\dots, N.
   \label{eq:dthetadt}
\ee
The heterogeneity of neurons results
from their excitability parameters $\eta_j$ being
chosen from a Lorentzian distribution
$$
\Lambda(\eta) = \frac{\Delta}{\pi} \frac{1}{(\eta - \eta_0)^2 + \Delta^2}
$$
with center $\eta_0$ and half-width at half maximum $\Delta$. In the absence of coupling
($\kappa=0$) a single theta neuron is the normal form of the saddle-node-on-invariant-circle
(SNIC) bifurcation~\cite{erm96,ermkop86}. The neuron is excitable if its intrinsic drive $\eta_j$ is negative,
and fires periodically with frequency $\sqrt{\eta_j}/\pi$ if $\eta_j>0$. 
The theta neuron is equivalent to a quadratic integrate-and-fire (QIF) neuron 
under the assumption of infinite firing threshold and reset values~\cite{monpaz15}.

Neurons interact synaptically by sensing input currents $I_j(t)$,
which consist of impulses coming
from all neurons in the system.
In addition, the coupling strength $\kappa$
characterizes the sensitivity of neurons
to these currents. (Coupling with gap junctions could also be 
included~\cite{lai15,omelai24}, but for simplicity we do not do so.)
To mimic the spatial organization
of the network, we assume that neurons
are arranged on a ring, i.e., a one-dimensional array
with periodic boundary conditions,
and that their interaction depends
on the distance between them.
A simple example of such a model was proposed in~\cite{lai14a},
where the input currents in~(\ref{eq:dthetadt}) had the form
\be
I_j(t) = \frac{2\pi}{N}\sum_{k=1}^N G_{jk} P(\theta_k(t)),
\label{Def:I_j:simple}
\ee
with a pulsatile function
$P(\theta)=(2/3)(1-\cos{\theta})^2$
representing the pulse of current emitted
by a neuron when $\theta$ increases through $\pi$.
(Note that, as explained in~\cite{lukbar13},
the exponent of $2$ in $P(\theta)$
can be varied to make the function more narrow.)
The coupling weights in~(\ref{Def:I_j:simple})
are defined by $G_{jk}=G(2\pi(j-k)/N)$
for some $2\pi$-periodic function $G(\cdot)$.
Thus, the resulting model~(\ref{eq:dthetadt}), (\ref{Def:I_j:simple}) represents $N$ theta neurons
equally-spaced around a ring of circumference~$2\pi$.
Studying a ring of neurons is appropriate when modeling behavior related to an angular
variable such as head direction~\cite{zha96} or angular distance in working memory
tasks~\cite{wimnyk14}. Other studies of ring networks include~\cite{laicho01,esnrox17,lai15,roxmon11}.

The spatio-temporal patterns
observed in model~(\ref{eq:dthetadt}), (\ref{Def:I_j:simple})
were analyzed in detail in~\cite{omelai22,laiome23}.
Roughly speaking, they include spatially uniform
quiescent and active states as well as bump states,
i.e., localized patches of activity,
which can be stationary, moving, or ``breathing'' (periodically varying width).
In the limit $N\gg 1$, the properties
of all the above states can be described
using a neural field equation
following from application of the Ott/Antonsen theory~\cite{ottant08,lukbar13}.
In this paper, we perform a similar analysis,
but in a much more complex setup,
where the interaction currents $I_j(t)$
contain physically motivated time-delays.
More specifically, we consider two types of delays.
\smallskip

(i) When impulses from other neurons arrive
at the $j$th neuron, it may exhibit some delay
before responding to them. Also,
the effect of the impulses is not necessarily instantaneous,
but may continue over time.
Both properties can be modeled by an input current with distributed delays
\be
I_j(t) = \frac{2\pi}{N}\sum_{k=1}^N G_{jk}\int_0^\infty
S(t')P(\theta_k(t-t'))\: dt'
= \frac{2\pi}{N}\sum_{k=1}^N G_{jk}\int_{-\infty}^{t}
S(t-s)P(\theta_k(s))\: ds,
\label{eq:Idisc}
\ee
where the delay kernel $S(t)$ is zero up until
a fixed time $\tau > 0$
and is a non-increasing function for $t \ge \tau$
such that $\int_0^\infty S(t) dt = 1$.
Such kernels have been studied in~\cite{lailon03},
for example.

For model~(\ref{eq:dthetadt}), (\ref{eq:Idisc}),
we distinguish two qualitatively different situations:
The case of a monotonically decreasing kernel $S(t)$
with infinite support, an example of which is
\be
   S(t) = \begin{cases} 0, & t<\tau, \\ a^{-2}(t-\tau)e^{-(t-\tau)/a}, & t\geq \tau \end{cases}
   \label{eq:kern}
\ee
with $a > 0$,
and the case of a constant finite-support kernel
\be
   S(t)=\begin{cases} \frac{1}{a} & \tau<t<\tau+a, \\ 0 & \mbox{otherwise} \end{cases}
   \label{eq:kern:finite}
\ee
with $a > 0$, as considered in, for example~\cite{rahbly15}. Note that in both cases
we have $\lim_{a\to 0} S(t) = \delta(t-\tau)$, i.e.,~just a discrete delay of $\tau$.
\smallskip

(ii) Another type of delay can occur in a neural network due to the finite speed of impulse transmission.
We call such delays conduction (or propagation) delays and represent them using input currents of the form
\be
   I_j(t)=\frac{2\pi}{N}\sum_{k=1}^N G_{jk}P\left(\theta_k\left(t-\tau-\frac{2\pi|j-k|}{Nc}\right) \right),\label{eq:IdiscB}
\ee
where $\tau$ is a fixed delay
and $c$ is the conduction speed.
Note that the expression $|j - k|$ here means
$\min( |j - k|, N - |j - k| )$ and thus
takes into account periodic boundary conditions.
Eqs.~\eqref{eq:dthetadt}, \eqref{eq:IdiscB}
model a system in which impulses arrive
after a time equal to a fixed delay $\tau$
(which may be zero) plus a time proportional
to the distance they have traveled.
Similar types of delays were studied in~\cite{setsen08,kanan24,zan00}, and
specifically in the context of neural models in~\cite{croerm97,vencoo07,velfau11,vanjan13,saycoo24}.

\section{Methods}

In the continuum limit case $N\to\infty$,
the dynamics of the discrete system~(\ref{eq:dthetadt}) admits a probabilistic description.
Assuming that $x_j = 2\pi j / N$ denotes
the position of the $j$th neuron in the interval $[0,2\pi]$,
we can describe the state of system~(\ref{eq:dthetadt})
with the probability density $\rho(\theta,x,\eta,t)$
such that $\rho(\theta,x,\eta,t) d\theta\: dx\: d\eta$
is the probability of finding neurons
with $\theta_j\in(\theta,\theta+d\theta)$,
$x_j\in(x,x+dx)$
and $\eta_j\in(\eta,\eta+d\eta)$ at time $t$.
Using the Ott/Antonsen theory~\cite{ottant08,ottant09},
it can be shown that the long-term dynamics
of the density $\rho(\theta,x,\eta,t)$
is attracted to an invariant manifold
parametrized by the complex-valued variable 
$$
z(x,t) = \int_{-\infty}^\infty
\int_0^{2\pi} e^{i \theta} \rho(\theta,x,\eta,t) d\theta\: d\eta
$$
known as the local order parameter~\cite{omelai22,lai14a},
which by definition satisfies $|z(x,t)|\le 1$.
Moreover, the dynamics on this manifold
is described by the neural field equation
\be
     \pf{z}{t} = \fr{[ i ( \eta_0 + \kappa I(x,t) ) - \Delta ] ( 1 + z )^2}{2} - i \fr{(1 - z)^2}{2},
\label{Eq:MeanField}
\ee
where $\eta_0$ and $\Delta$ are parameters of the Lorentzian distribution $\Lambda(\eta)$
and $I(x,t)$ is an integral term corresponding
to the mean-field synaptic coupling in the discrete system~(\ref{eq:dthetadt}).
In particular, for model~(\ref{eq:dthetadt}), (\ref{eq:Idisc}) we have
\be
     I(x,t) = \int_{-\pi}^\pi G(x - y) \int_{-\infty}^{t}
S(t-s)H\left( z\left(y,s\right) \right) ds\: dy \label{eq:Iint}
\ee
where (for the particular function $P(\theta)$ we use)
\be
   H(z) = 1 - \fr{2}{3} ( z + \overline{z} ) + \fr{1}{6} ( z^2 + \overline{z}^2 )\equiv 1+2\:\Real{[D(z)]}
\label{Def:H}
\ee
and $D(z)=-2z/3+z^2/6$.
On the other hand, for model ~(\ref{eq:dthetadt}),
(\ref{eq:IdiscB}) we have
\be
I(x,t) = \int_{-\pi}^\pi G(x - y) 
H\left( z\left(y,t -\tau-\fr{|x - y|}{c} \right) \right)  dy
\label{eq:IintC}
\ee
where $H(z)$ is again given by~(\ref{Def:H}).
Note that due to periodic boundary conditions
the expression $|x - y|$ in~(\ref{eq:IintC})
means $\min( |x- y|, 2\pi - |x - y| )$,
or equivalently $\arccos( \cos(x-y) )$.

\begin{remark}
In the case of the delay kernel~(\ref{eq:kern}),
we can use the linear chain trick~\cite{an80,hurkir19}
to rewrite the system~(\ref{Eq:MeanField}), (\ref{eq:Iint}) as one with only a discrete delay,
rather than a continuum of delays.
Indeed, for $t > \tau$, we have
\[
   S'(t)=a^{-2}e^{-(t-\tau)/a}-a^{-3}(t-\tau)e^{-(t-\tau)/a}=a^{-2}e^{-(t-\tau)/a}-S(t)/a.
\]
Thus
\be
   \pf{I(x,t)}{t} = \int_{-\pi}^\pi G(x-y) \int_{-\infty}^{t}S'(t-s)H\left( z\left(y,s\right) \right) ds\: dy = [Y(x,t)-I(x,t)]/a
\label{eq:dIdt}
\ee
where
\[
   Y(x,t) \equiv \frac{1}{a}\int_{-\pi}^\pi G(x-y)\int_{-\infty}^{t-\tau}e^{-(t-\tau-s)/a}H\left( z\left(y,s\right) \right) ds\: dy.
\]
Differentiating $Y$ with respect to time we find
\be
   \pf{Y(x,t)}{t}=\frac{1}{a}\left[\int_{-\pi}^\pi G(x-y) H\left( z\left(y,t-\tau\right) \right)dy-Y(x,t)\right]. \label{eq:dydt}
\ee
Thus we have written system~(\ref{Eq:MeanField}), (\ref{eq:Iint}) with the delay kernel~(\ref{eq:kern})
as two non-delayed equations, \eqref{Eq:MeanField} and~\eqref{eq:dIdt},
and one with a discrete delay,~\eqref{eq:dydt}. If we had a higher power of $t-\tau$
multiplying the exponential in~\eqref{eq:kern} we could introduce more intermediate variables
like $Y$, still reducing the original equation to a finite set of non-delayed equations
and one with a discrete delay~\cite{lailon03}, but for the purpose of explaining our
technique we use the kernel~\eqref{eq:kern}.
\label{Remark:ChainTrick}
\end{remark}

\begin{remark}
In the case of the finite-support
kernel~(\ref{eq:kern:finite}),
the integral term~(\ref{eq:Iint})
can be written explicitly
\be
I(x,t) = \frac{1}{a}\int_{-\pi}^\pi G(x - y) \int_{\tau}^{\tau+a}
H\left( z\left(y,t-s\right) \right) ds\: dy.
\label{eq:IintA}
\ee
Obviously, system~(\ref{Eq:MeanField}), (\ref{eq:IintA})
has a continuum of delays but they do not extend infinitely far in the past.
\end{remark}

To understand pattern formation processes
in discrete theta neuron systems~(\ref{eq:dthetadt})
with different types of coupling delay,
we analyze the neural field equation~(\ref{Eq:MeanField})
with the corresponding input current terms.
More specifically, system~(\ref{Eq:MeanField}), (\ref{eq:Iint}) with delay kernel~(\ref{eq:kern})
is used to study the activity patterns
in model~(\ref{eq:dthetadt}), (\ref{eq:Idisc})
with distributed delays having infinite support; see Sec.~\ref{sec:inf}.
Analogously, we use system~(\ref{Eq:MeanField}), (\ref{eq:IintA}) to study the activity patterns
in model~(\ref{eq:dthetadt}), (\ref{eq:Idisc})
with distributed delays having finite support; see Sec.~\ref{sec:compact}.
Finally, model~(\ref{eq:dthetadt}),
(\ref{eq:IdiscB}) with conduction delays
is analyzed using system~(\ref{Eq:MeanField}),
(\ref{eq:IintC}); see Sec.~\ref{sec:cond}.

We consider each of the above continuum limit systems
with the simplest non-constant coupling function
\be
G(x)=A+B\cos{x},
\label{Def:G}
\ee
where $A$ and $B$ are constants. Depending on their values, the coupling could be purely
excitatory, purely inhibitory, or ``Mexican hat'' (positive for small $|x|$ and
negative for large $|x|$)~\cite{erm98,roxbru05}. Our general analysis does not rely
on having only one harmonic in $G(x)$, but including more harmonics would increase
the computational effort commensurately.
Using our previous results from~\cite{omelai22},
we choose the parameters $A$, $B$, $\kappa$, $\eta_0$ and~$\Delta$ such that in model~(\ref{Eq:MeanField})
with a non-delayed input current
\be
I(x,t) = \int_{-\pi}^\pi G(x - y) H( z(y,t) )  dy,
\label{eq:Iint0}
\ee
a spatially uniform steady state
and a stationary bump state coexist stably.
Then we show how linear stability analysis
can be performed for these time-independent states 
in the presence of delayed input current terms.
This allows us to identify a number of Hopf bifurcation
curves on the $(\tau,a)$-plane for the input current~(\ref{eq:Iint}) with kernels~(\ref{eq:kern})
and~(\ref{eq:kern:finite}),
and on the $(\tau,c)$-plane
for the input current~(\ref{eq:IintC}).
We observe that Hopf bifurcations of a spatially uniform steady state typically give rise to traveling waves that represent moving bumps
in the corresponding discrete theta neuron systems.
(Note that in the degenerate case,
a spatially uniform periodic solution
may appear instead of a traveling wave.)
On the other hand, Hopf bifurcations of stationary bump states typically give rise to periodic solutions
of neural field equations
that represent breathing bumps.
To calculate entire branches of traveling waves
and periodic spatially nonuniform solutions,
we derive self-consistency equations
that significantly speed up these calculations.
As a result, we obtain a global picture of the influence
of different delays on pattern formation processes
in spatially extended networks of theta neurons.

Note that since the values of stationary states do not depend on the delay parameters that
we vary, we do not expect to observe stationary bifurcations such as a saddle-node bifurcation
as delay parameters are varied. While all of our analysis involves the continuum limit equation~\eqref{Eq:MeanField},
we show some solutions of the discrete network~\eqref{eq:dthetadt} in Appendix~\ref{sec:disc}.
Before considering the effects of delays we briefly discuss the behavior without
delays.

\section{System without delays}
\label{sec:old}

The neural field equation~\eqref{Eq:MeanField}
with a non-delayed input current~\eqref{eq:Iint0}
was analyzed in~\cite{omelai22}.
In that paper, it was shown that Eqs.~\eqref{Eq:MeanField}, \eqref{eq:Iint0}
have various steady state solutions,
which can be represented in the form $Z(x) = U(w(x),\Delta)$, where 
\[
   w(x)=\eta_0+\kappa\int_{-\pi}^\pi G(x-y)H(Z(y)) dy
\]
is the corresponding spatial profile
of input current to neurons at position $x$, and
\[
   U(w,\Delta)=\frac{1-\sqrt{w+i\Delta}}{1+\sqrt{w+i\Delta}}
\]
with
\[
   \sqrt{w+i\Delta}=\frac{1}{\sqrt{2}}\left(\sqrt{\sqrt{w^2+\Delta^2}+w}+i\sqrt{\sqrt{w^2+\Delta^2}-w}\right).
\]
It is easy to see that for self-consistency
$w(x)$ must satisfy the following equation:
\be
  w(x)=\eta_0+\kappa\int_{-\pi}^\pi G(x-y)H(U(w(y),\Delta)) dy. \label{eq:self}
\ee
Using the form of $G(x)$, see~\eqref{Def:G},
and the translational invariance of the system, we can write
\[
   w(x)=\hat{w}_0+\hat{w}_1\cos{x}
\] 
for some constants $\hat{w}_0$ and $\hat{w}_1$,
and inserting this into~\eqref{eq:self} and equating the constant terms and the terms proportional
to $\cos{x}$ we find 
\begin{eqnarray}
    \hat{w}_0 & = & \eta_0+2\pi A\kappa\langle H(U(\hat{w}_0+\hat{w}_1\cos{x},\Delta))\rangle,
    \label{eq:selfB} \\[2mm]
    \hat{w}_1 & = & 2\pi B\kappa\langle H(U(\hat{w}_0+\hat{w}_1\cos{x},\Delta))\cos{x}\rangle,
    \label{eq:selfC}
\end{eqnarray}
where angle brackets indicate an average over $[-\pi,\pi]$.

Eqs.~\eqref{eq:selfB}--\eqref{eq:selfC}
are a pair of equations for the two unknowns $\hat{w}_0$ and $\hat{w}_1$.
Obviously, for $\hat{w}_1=0$, Eq.~\eqref{eq:selfC}
is always satisfied and Eq.~\eqref{eq:selfB}
can be used to obtain a $\hat{w}_0$-parametric representation of spatially uniform states.
For the parameter values given in Table~\ref{tab:param} there are three spatially uniform
states, $Z_1\approx 0.6024 - 0.7101i,Z_2\approx 0.5239-0.0753i$ and $Z_3\approx 0.0795 - 0.0138i$. 
For these the mean firing frequencies are $f_1\approx 0.0138,f_2\approx 0.0984$ and
$f_3\approx 0.2713$, where the frequency is given by~\cite{lai15,monpaz15}
\be
   f=\frac{1}{\pi}\Real\left(\frac{1-\overline{z}}{1+\overline{z}}\right). \label{eq:freq}
\ee
Applying the linear stability analysis
developed in~\cite{omelai22}, it can be shown that
$Z_1$ and $Z_3$ are stable, and $Z_2$ is unstable.

\begin{table}
\footnotesize
\caption{Parameter values}
\begin{center}
\begin{tabular}{|c||ccccccccccc|} \hline
Parameter\vphantom{$\int_0^1$} & & $\eta_0$ & & $\Delta$ & & $\kappa$ & & $A$ & & $B$ & \\ \hline
Value\vphantom{$\int_0^1$} & & $-0.4$ & & 0.04 & & 2 & & 0.1 & & 0.3 & \\ \hline
\end{tabular}
\end{center}
\label{tab:param}
\end{table}

The solutions of Eqs.~\eqref{eq:selfB}--\eqref{eq:selfC}
with $\hat{w}_1\neq 0$ describe steady state bump solutions of~\eqref{Eq:MeanField}, \eqref{eq:Iint0}.
For the parameters in Table~\ref{tab:param} there is a stable bump solution of~\eqref{Eq:MeanField}, \eqref{eq:Iint0} corresponding to $(\hat{w}_0,\hat{w}_1)\approx(0.4246,0.9282)$
and an unstable bump corresponding to $(\hat{w}_0,\hat{w}_1)\approx(0.6336,0.4171)$.

In the following sections, we will consider what happens
to several of the above stationary states in the presence
of delayed terms in Eq.~\eqref{Eq:MeanField}.

\section{Distributed delays --- infinite support}
\label{sec:inf}

In this section, we consider the neural field equation~(\ref{Eq:MeanField}) with an input current term
containing distributed delays~(\ref{eq:Iint}).
More specifically, we analyze the case
of a delay kernel~(\ref{eq:kern}).
Another kernel with finite support~(\ref{eq:kern:finite})
will be analyzed in Sec.~\ref{sec:compact}.
We perform the following steps.
First, we conduct a linear stability analysis
of stationary states in Eqs.~(\ref{Eq:MeanField}), (\ref{eq:Iint}).
This helps us to identify the bifurcations
corresponding to the onset of traveling wave solutions
and periodic solutions with spatial structure.
For each of these types of solution, we derive a self-consistency equation
and use it to efficiently find the entire solution branch.

\subsection{Stability of stationary states} 
\label{sec:stabinf}

We start by considering the stability of stationary states,
generalizing the results in~\cite{omelai22}.
Suppose that Eqs.~(\ref{Eq:MeanField}), (\ref{eq:Iint})
have a time-independent solution $Z(x)$.
Letting $z(x,t)=Z(x)+v(x,t)$, to linear order in $v$ we have
\begin{eqnarray}
\frac{\partial v}{\partial t} &=& \mu(x)v
+ i \kappa \zeta(x) \label{eq:lin}\\
&\times& \int_{-\pi}^\pi G(x-y)\int_{-\infty}^{t}S(t-s)
\left[D'(Z(y))v(y,s)+\overline{D'(Z(y))}\overline{v(y,s)}\right]ds\ dy,
\nonumber
\end{eqnarray}
where
\be
  \mu(x)=\left[i\left(\eta_0+\kappa\int_{-\pi}^\pi G(x-y)H(Z(y))\:dy\right)-\Delta\right](1+Z(x))+i(1-Z(x))
  \label{Def:mu}
\ee
and
\be
\zeta(x) = (1+Z(x))^2/2
\label{Def:zeta}
\ee
and
\be
D'(z) = -2/3 + z/3.
\label{Def:D_}
\ee
Letting
\[
   v(x,t)=v_+(x)e^{\lambda t}+\overline{v}_-(x)e^{\overline{\lambda}t}
\]
and evaluating the integral over $s$ in~\eqref{eq:lin} we obtain
\begin{eqnarray*}
\lambda v_+(x)e^{\lambda t} &+& \overline{\lambda}\overline{v}_-(x)e^{\overline{\lambda}t}
= \mu(x)\left[v_+(x)e^{\lambda t}+\overline{v}_-(x)e^{\overline{\lambda}t}\right] \\
&+& i \kappa \zeta(x) \int_{-\pi}^\pi G(x-y)\left[D'(Z(y))\left\{v_+(y)\frac{e^{\lambda (t-\tau)}}{(1+a\lambda)^2}+\overline{v}_-(y)\frac{e^{\overline{\lambda} (t-\tau)}}{(1+a\overline{\lambda})^2}\right\} \right.\\
&+& \left.\overline{D'(Z(y))}\left\{\overline{v}_+(y)\frac{e^{\overline{\lambda} (t-\tau)}}{(1+a\overline{\lambda})^2}+v_-(y)\frac{e^{\lambda (t-\tau)}}{(1+a\lambda)^2}\right\}\right]dy.
\end{eqnarray*}
Then separately equating the terms in
$e^{\lambda t}$ and $e^{\overline{\lambda} t}$ we obtain
\be
   \lambda\bp v_+ \\ v_- \ep = \mathcal{L}(\lambda)\bp v_+ \\ v_- \ep,
\label{eq:stabA}
\ee
where
\begin{eqnarray*}
\mathcal{L}(\lambda)\bp v_+ \\ v_- \ep &=& \bp \mu(x) & 0 \\ 0 & \overline{\mu(x)} \ep \bp v_+ \\ v_- \ep \\
&+& \frac{i\kappa e^{-\lambda\tau}}{(1+a\lambda)^2}
\bp \zeta(x) & \zeta(x) \\ - \overline{\zeta(x)} & -\overline{\zeta(x)} \ep 
\bp \int_{-\pi}^\pi G(x-y)D'(Z(y))v_+(y)\: dy \\ \int_{-\pi}^\pi G(x-y)\overline{D'(Z(y))}v_-(y)\: dy \ep.
\end{eqnarray*}
Here, as usual, we say that $\lambda\in\mathbb{C}$
is the spectral value of Eq.~\eqref{eq:stabA} if the operator $\lambda\mathcal{I} - \mathcal{L}(\lambda)$,
where $\mathcal{I}$ is the identity operator,
is not invertible. If, in addition, $\lambda\mathcal{I} - \mathcal{L}(\lambda)$ is a Fredholm operator
of index zero, we say that $\lambda$ belongs
to the discrete spectrum $\sigma_\mathrm{discr}$,
while otherwise $\lambda$ belongs to the essential spectrum $\sigma_\mathrm{ess}$.

The essential spectrum corresponding to $Z(x)$
is determined by the first term
in the definition of $\mathcal{L}(\lambda)$
and has the form
\[
   \sigma_\mathrm{ess}=\{\mu(x): x\in[0,2\pi]\} \cup \{\mathrm{c.c.}\},
\]  
where $\{\mathrm{c.c.}\}$ is the complex conjugate of the previous term. As for the discrete spectrum,
it can be found as follows.
Because of the form of $G(x)$, see~\eqref{Def:G}, we can write
\be
   \bp \int_{-\pi}^\pi G(x-y)D'(Z(y))v_+(y)\: dy \\ \int_{-\pi}^\pi G(x-y)\overline{D'(Z(y))}v_-(y)\: dy \ep =V_1+V_2\cos{x}+V_3\sin{x} \label{eq:formofG}
\ee
where $V_{1,2,3}\in\mathbb{C}^2$. Substituting~\eqref{eq:formofG} into~\eqref{eq:stabA} we have
\[
   \bp \lambda-\mu(x) & 0 \\ 0 & \lambda-\overline{\mu(x)} \ep \bp v_+ \\ v_- \ep =
\frac{i\kappa e^{-\lambda\tau}}{(1+a\lambda)^2}
\bp \zeta(x) & \zeta(x) \\ - \overline{\zeta(x)} & -\overline{\zeta(x)} \ep
\left[V_1+V_2\cos{x}+V_3\sin{x}\right]
\]
and thus, assuming $\lambda\notin\sigma_\mathrm{ess}$, solving this equation for $v_+$ and $v_-$ and multiplying on the 
left by a diagonal matrix containing the derivatives of $D$ we have
\be
   \bp D'(Z(x)) v_+ \\ \overline{D'(Z(x))}v_- \ep = L(x,\lambda)\left[V_1+V_2\cos{x}+V_3\sin{x}\right] \label{eq:vv}
\ee
where
\[
   L(x,\lambda)=\bp D'(Z(x)) & 0 \\ 0 & \overline{D'(Z(x))} \ep \bp \lambda-\mu(x) & 0 \\ 0 & \lambda-\overline{\mu(x)} \ep^{-1} \frac{i\kappa e^{-\lambda\tau}}{(1+a\lambda)^2}
\bp \zeta(x) & \zeta(x) \\ - \overline{\zeta(x)} & -\overline{\zeta(x)} \ep
\]
is a $2\times 2$ matrix. Writing $y$ instead of $x$ in~\eqref{eq:vv}, then multiplying 
both sides by $G(x-y)$ and integrating $y$ from $-\pi$ to $\pi$, and using~\eqref{eq:formofG}
and the form of $G(x)$, we obtain
\begin{align*}
   V_1+V_2\cos{x}+V_3\sin{x} & = \int_{-\pi}^\pi G(x-y)L(y,\lambda)\left[V_1+V_2\cos{y}+V_3\sin{y}\right] dy \\
   & =  A\int_{-\pi}^\pi L(y,\lambda)\left[V_1+V_2\cos{y}+V_3\sin{y}\right] dy \\
   & +  B\cos{x}\int_{-\pi}^\pi L(y,\lambda)\cos{y} \left[V_1+V_2\cos{y}+V_3\sin{y}\right] dy \\
   & +  B\sin{x}\int_{-\pi}^\pi L(y,\lambda)\sin{y} \left[V_1+V_2\cos{y}+V_3\sin{y}\right] dy.
\end{align*}
Matching coefficients of the constant, $\cos{x}$ and $\sin{x}$ terms, 
this gives a set of three simultaneous homogeneous linear equations for $V_{1,2,3}$
that will have a non-zero solution if the determinant of $I_6-C(\lambda)$ is zero, where
$I_6$ is the $6\times 6$ identity matrix and 
\[
   C(\lambda)=2\pi \bp A\la L(x,\lambda) \ra & A\la L(x,\lambda)\cos{x} \ra & A\la L(x,\lambda)\sin{x} \ra \\
    B \la L(x,\lambda)\cos{x} \ra & B\la L(x,\lambda)\cos^2{x} \ra & B\la L(x,\lambda)\cos{x}\sin{x} \ra \\ 
    B \la L(x,\lambda)\sin{x} \ra & B\la L(x,\lambda)\sin{x}\cos{x} \ra & B\la L(x,\lambda)\sin^2{x} \ra \ep
\]
and angle brackets indicate an average over $[-\pi,\pi]$, applied to each component
of a matrix separately. Solutions of the equation
\be
   |I_6-C(\lambda)|=0
\label{Eq:sigma:discr}
\ee
determine the discrete spectrum $\sigma_\mathrm{disc}$
associated with the stability of the stationary state $Z(x)$.

\begin{remark}
If $Z(x)$ is even, i.e. $Z(-x) = Z(x)$, then $L(x,\lambda)$ is also even, and thus some components of $C(\lambda)$ are zero, giving
\[
   C(\lambda) =\bp C_{11}(\lambda) & C_{12}(\lambda) & 0 \\ C_{21}(\lambda) & C_{22}(\lambda) & 0 \\ 0 & 0 & C_{33}(\lambda) \ep.
\]
In this case, Eq.~\eqref{Eq:sigma:discr} is decomposed
into two lower dimensional equations
\[
   \left|I_4-\bp C_{11}(\lambda) & C_{12}(\lambda) \\ C_{21}(\lambda) & C_{22}(\lambda) \ep\right|=0
\]
and
\[
   |I_2-C_{33}(\lambda)|=0.
\]
\label{Remark:Z:Even}
\end{remark}

\begin{remark}
If $Z(x)$ is constant, then $L(x,\lambda)$ also does not depend on $x$, and therefore
\[
   C(\lambda) = \bp 2 \pi A L(\lambda) & 0 & 0 \\ 0 & \pi B L(\lambda) & 0 \\ 0 & 0 & \pi B L(\lambda) \ep.
\]
In this case, Eq.~\eqref{Eq:sigma:discr} is decomposed
into two lower dimensional equations
\be
   |I_2-2 \pi A L(\lambda)|=0
\label{eq:eigA}
\ee
and
\be
   |I_2-\pi B L(\lambda)|=0.
\label{eq:eigB}
\ee
(Note that Eq.~\eqref{eq:eigB} appears in two identical copies.)
Importantly, only eigenvalues from Eq.~\eqref{eq:eigB}
are associated with an instability to a spatially non-uniform state. The corresponding bifurcation
leads to the creation of traveling waves. There are no bifurcations to stationary patterns, as explained in Appendix~\ref{sec:noturing}.
\label{Remark:Z:Constant}
\end{remark}

Every stationary state of Eqs.~\eqref{Eq:MeanField}, \eqref{eq:Iint0}
is simultaneously a stationary state of Eqs.~(\ref{Eq:MeanField}), (\ref{eq:Iint}),
but its stability may be different
for different types of kernel $S(t)$.
For example, for a spatially uniform state $Z$
its stability is determined by the position
of eigenvalues given by the characteristic
equations~\eqref{eq:eigA} and~\eqref{eq:eigB}.

Consider the two spatially uniform states~$Z_1$ and~$Z_3$,
presented in Section~\ref{sec:old},
which are stable in the absence of delays
(i.e. for $\tau = 0$ and $a\to 0$).
Varying $a$ and $\tau$ we find that $Z_1$ never loses stability, but $Z_3$ does, 
via its discrete spectrum. Thus we
concentrate on this state.
To find solutions of~\eqref{eq:eigB} we evaluate $\Real (I_2-\pi BL(\lambda))$
and $\Imag (I_2-\pi BL(\lambda))$ for a range of values of $\lambda$ in the complex plane.
We plot the zero contours of both $\Real (I_2-\pi BL(\lambda))$ and $\Imag (I_2-\pi BL(\lambda))$ ---
solutions of~\eqref{eq:eigB} are where these curves intersect. To find where $Z_3$ loses
stability we solve the three equations $\Real (I_2-\pi BL(\lambda))=0$,
$\Imag (I_2-\pi BL(\lambda))=0$ and $\Real(\lambda)=0$, where the three unknowns are
$\Real(\lambda),\Imag(\lambda)$ and a parameter such as $\tau$.

Following where the eigenvalues
given by~\eqref{eq:eigB}, evaluated at $Z=Z_3$, are a purely imaginary pair 
gives the blue curves in Fig.~\ref{fig:atau}. These curves have some form of self-similarity.
Indeed, if $|I_2-\pi BL(i\omega)|=0$ for some delay $\tau=\hat{\tau}$, where $\omega$ is real,
then $|I_2-\pi BL(i\omega)|=0$ when the delay $\tau=\hat{\tau}+2\pi n/\omega$ for any
integer $n$, i.e.~this instability will reoccur at infinitely many values of the delay.
So given one of the blue curves shown in Fig.~\ref{fig:atau}, all of them can be drawn
(assuming that the imaginary part of the eigenvalue along a curve is known).
This is the same phenomenon as the reappearance of periodic solutions in systems with
one delay~\cite{yanper09}: if a periodic solution with period $T_0$ occurs when the delay
is $\tau_0$, then the same periodic solution occurs when the delay is $\tau_0+nT_0$ for any
integer $n$~\cite{laikra25}.

(We also followed where the eigenvalues
given by~\eqref{eq:eigA}, evaluated at $Z=Z_3$, are a purely imaginary pair. These curves lie in
the white region of Fig.~\ref{fig:atau}, where $Z_3$ is already unstable, so are not
associated with the creation of stable solutions.) 

%%%%%%%%%
\begin{figure}[htbp]
\centering
\includegraphics[width=0.8\textwidth]{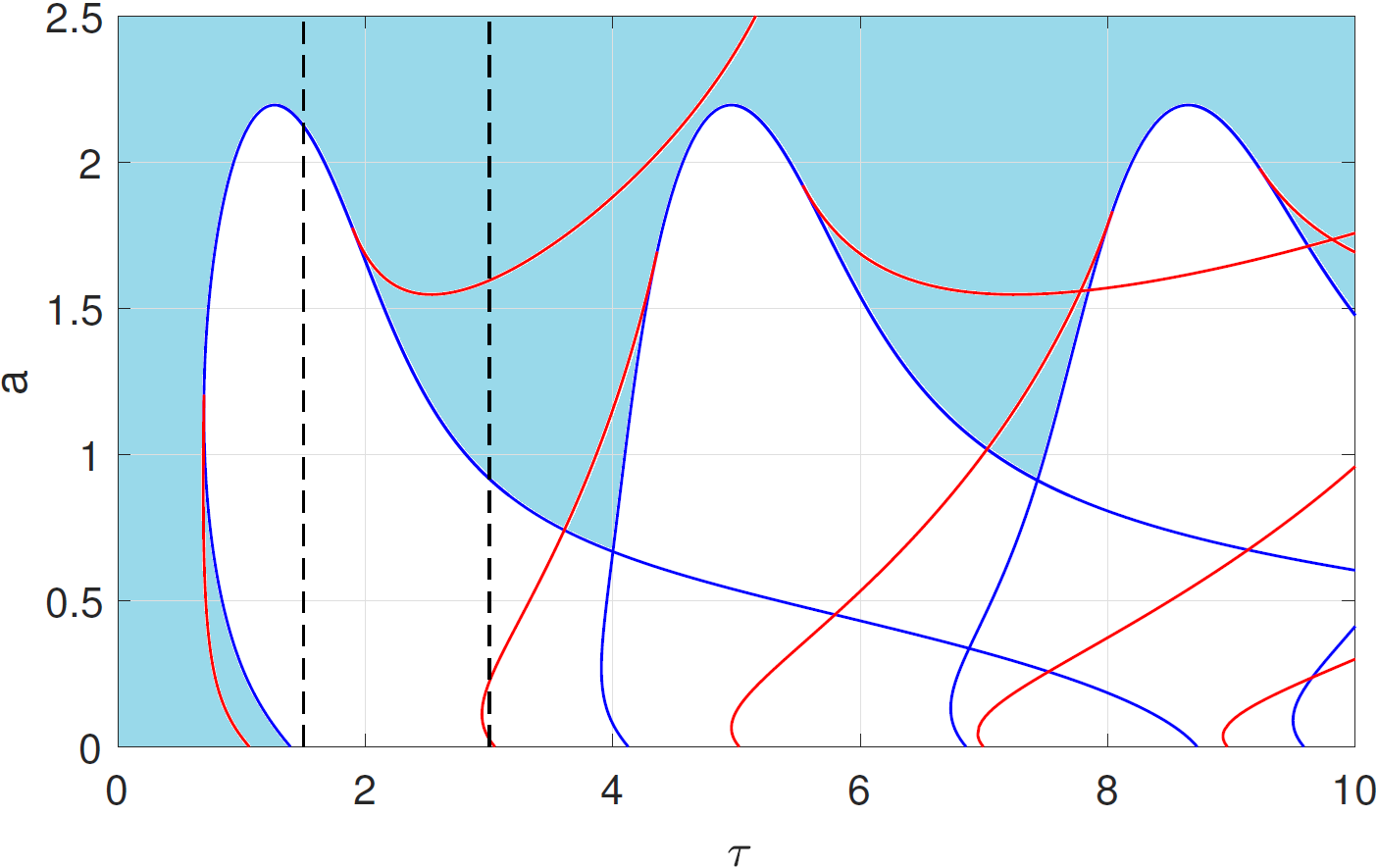}
\caption{Blue: Hopf bifurcation (for eigenvalues satisfying~\eqref{eq:eigB}) of the spatially uniform state $Z_3$.
Red: saddle-node bifurcations of traveling wave solutions of~\eqref{Eq:MeanField}--\eqref{eq:Iint}.
 The spatially uniform state $Z_3$ is stable in the blue shaded region. The dashed vertical lines refer
 to Fig.~\ref{fig:varya}.
Parameters as in Table~\ref{tab:param}.
}
\label{fig:atau}
\end{figure}
%%%%%%%%%%%%%%

Next, we examine the stability of the stable bump
from Section~\ref{sec:old}
as delay parameters $a$ and $\tau$ are varied.
Using Remark~\ref{Remark:Z:Even},
we decompose the characteristic equation~\eqref{Eq:sigma:discr}
into two lower dimensional equations.
As above, we can numerically find the eigenvalues in the discrete spectrum
associated with the stability of a bump,
and following Hopf bifurcations of the stable bump,
we obtain the blue curves in Fig.~\ref{fig:hopfbump}.
As in Fig.~\ref{fig:atau}, the curves of Hopf bifurcations are mapped to one another under a transformation involving the imaginary part of the eigenvalues.

%%%%%%%%%
\begin{figure}[htbp]
\centering
\includegraphics[width=0.8\textwidth]{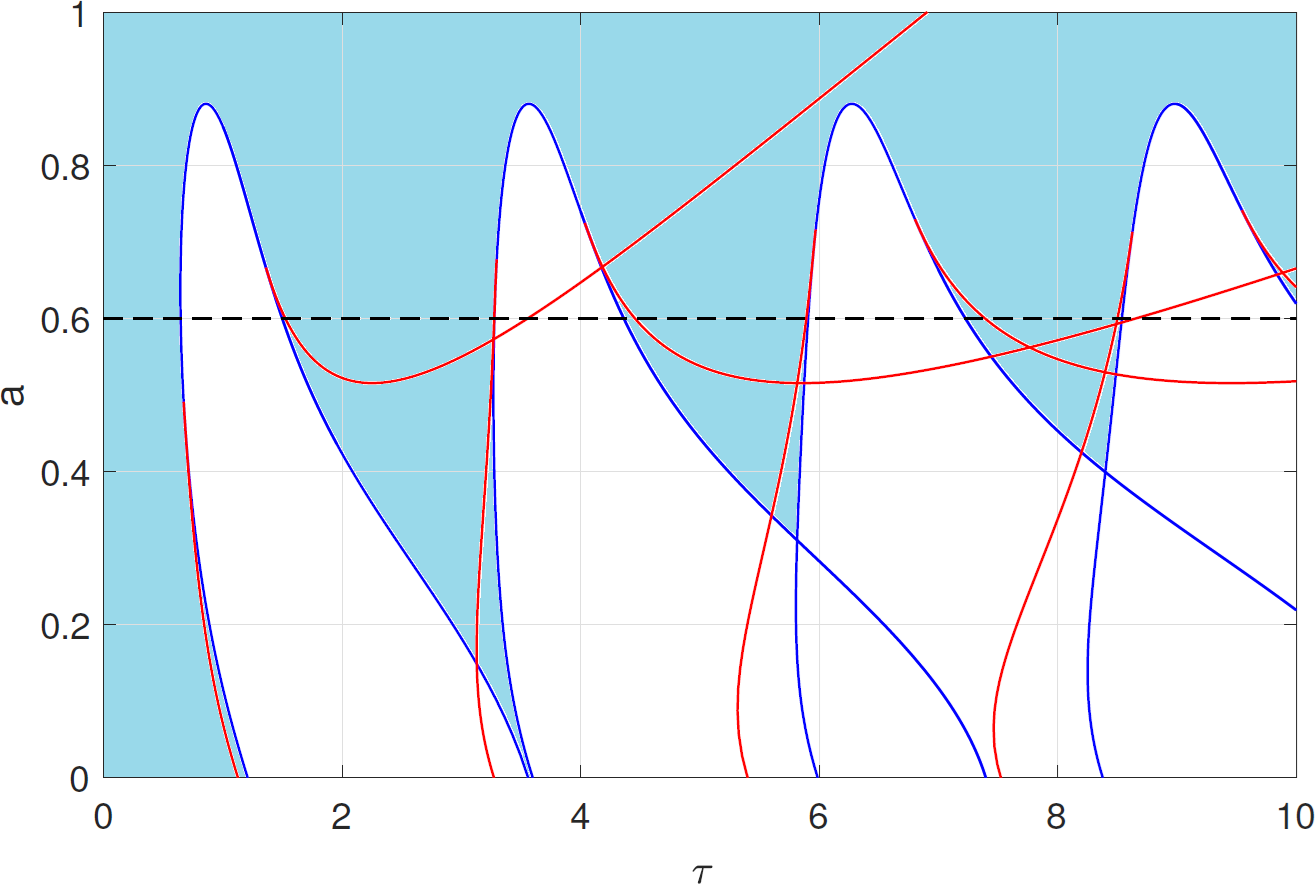}
\caption{Blue: Hopf bifurcations of the stable stationary bump. The bump is stable in the blue shaded region.
 Red curves show saddle-node bifurcations of periodic
solutions. The dashed black line refers to Fig.~\ref{fig:vtau}.
Parameters as in Table~\ref{tab:param}.
}
\label{fig:hopfbump}
\end{figure}
%%%%%%%%%%%%%%

\subsection{Traveling waves}
\label{sec:travinf}

Each Hopf bifurcation shown in Fig.~\ref{fig:atau} corresponds
to the emergence of a traveling wave solution
of Eqs.~(\ref{Eq:MeanField}), (\ref{eq:Iint})
from the same spatially uniform state.
To show this and to find the entire branch
of traveling waves, we use a special semi-analytical approach.

First, we recall that according to Remark~\ref{Remark:ChainTrick},
system~(\ref{Eq:MeanField}), (\ref{eq:Iint})
with delay kernel~(\ref{eq:kern})
is equivalent to the three equations~\eqref{Eq:MeanField}, \eqref{eq:dIdt} and~\eqref{eq:dydt}.
For a traveling wave moving to the left  with a fixed speed $v$,
the profile of the wave in the comoving frame $\xi = x + v t$
is given by a $2\pi$-periodic solution of
\begin{equation}
v\df{z}{\xi} = \fr{[ i ( \eta_0 + \kappa I(\xi)) - \Delta ] ( 1 + z )^2}{2} - i \fr{(1 - z)^2}{2}
\label{Eq:MeanFieldA}
\end{equation}
where
\be
   v\fr{dI}{d\xi}=(Y-I)/a \label{eq:dIdeta}
\ee
and
\be
   v\fr{dY}{d\xi}=\frac{1}{a}\left[\int_{-\pi}^\pi G(\xi-y) H\left( z\left(y-v\tau\right) \right)dy-Y\right]. \label{eq:dYdeta}
\ee
Using the form of $G(x)$, see~\eqref{Def:G}, we can write~\eqref{eq:dYdeta} as
\begin{eqnarray}
v\fr{dY}{d\xi} &=& \frac{1}{a}\left[A\int_{-\pi}^\pi H( z(y-v\tau))dy
+ B\cos{\xi}\int_{-\pi}^\pi H( z(y-v\tau))\cos{y}\:dy \right. \label{eq:dYdetaA} \\
&+& \left. B\sin{\xi}\int_{-\pi}^\pi H( z(y-v\tau))\sin{y}\:dy-Y\right], \nonumber
\end{eqnarray}
the $2\pi$-periodic solution of which is 
$Y(\xi)=\alpha+\beta\cos{\xi}+\gamma\sin{\xi}$
for some constants $\alpha,\beta$ and $\gamma$ 
that depend in a simple explicit way on the values of the integrals
in~\eqref{eq:dYdetaA}. Given this form of $Y$, the $2\pi$-periodic solution
of~\eqref{eq:dIdeta} is $I(\xi)=\delta+\epsilon\cos{\xi}+\phi\sin{\xi}$
for some constants $\delta,\epsilon$ and $\phi$ that depend in a simple and explicit way
on the coefficients in $Y$.

An important ingredient of the following analysis
is a special property of Eq.~(\ref{Eq:MeanFieldA}), described in~\cite{ome23}:
for any $v, \Delta > 0$, any $2\pi$-periodic coefficient $I(\xi)$,
and any values of the coefficients $\eta_0$ and $\kappa$,
Eq.~(\ref{Eq:MeanFieldA}) has one $2\pi$-periodic solution $z(\xi)$
such that $|z(\xi)| < 1$. This solution is not explicitly known,
but it can be easily computed using the fact
that~(\ref{Eq:MeanFieldA}) is a periodically-forced Riccati equation.
Put briefly, one numerically integrates~\eqref{Eq:MeanFieldA} for $0\leq \xi\leq 2\pi$ 
three times with three different initial conditions.
The three final states are enough to determine the initial condition
of~\eqref{Eq:MeanFieldA} for which the solution is $2\pi$-periodic
(and less than one in magnitude), and this
equation is then integrated a fourth time with this initial condition
to construct the periodic solution, $z(\xi)$.

From the above explanation, it is clear that if we are only interested
in $2\pi$-periodic solutions of~\eqref{Eq:MeanFieldA}--\eqref{eq:dYdeta}
such that $|z(\xi)| < 1$, it is sufficient to consider only $I(\xi)$ as the unknown function,
while the other unknowns $z(\xi)$ and $Y(\xi)$ can be reconstructed from it.
Moreover, for the coupling function~$G(x)$ given by~\eqref{Def:G},
the unknown function~$I(\xi)$ is specified by three scalars,
which significantly reduces the complexity of the problem.
Putting everything together, we have an algorithm
for describing traveling wave solutions of~\eqref{Eq:MeanField}, \eqref{eq:Iint}
with delay kernel~\eqref{eq:kern} in a self-consistent way:
\begin{enumerate}
\item Given values of $\delta,\epsilon$ and $\phi$ describing the profile of $I(\cdot)$ in a snapshot
of a traveling wave solution of~\eqref{Eq:MeanField} for some set of parameter values, 
construct $I(\xi)$ and
find the $2\pi$-periodic solution $z(\xi)$ of~\eqref{Eq:MeanFieldA} with $|z(\xi)|<1$.

\item Given this $z(\xi)$, calculate $H(z(\xi))$ and write it as
\[
   H(z(\xi))=\Omega+\Phi\cos{(\xi)}+\Psi\sin{(\xi)}+\mbox{higher harmonics}
\]
for some known constants $\Omega,\Phi$ and $\Psi$, and perform the integrals
in~\eqref{eq:dYdetaA}.

\item Once the integrals
in~\eqref{eq:dYdetaA} are known, calculate the constants describing $Y(\xi)$: 
$\alpha,\beta$ and $\gamma$.

\item Since $Y(\xi)$ is now known, substitute it into~\eqref{eq:dIdeta} and solve
for $I(\xi)$, which will be of the
form $I(\xi)=\tilde{\delta}+\tilde{\epsilon}\cos{\xi}+\tilde{\phi}\sin{\xi}$.
If $\tilde{\delta}=\delta,  \tilde{\epsilon}=\epsilon$ and $\tilde{\phi}=\phi$ we are done.
\end{enumerate}
Note that due to the translational invariance of the system we can set $\epsilon=0$.
In practice we solve 
\[
   \delta-\tilde{\delta}=0 \qquad \mbox{ and } \qquad \phi-\tilde{\phi}=0
\]
along with the pinning condition $\tilde{\epsilon}=0$, 
giving three equations in three unknowns ($\delta,\phi$ and $v$).
These equations are solved numerically using Newton's method, and solutions can be followed
as parameters are varied using pseudo-arclength continuation~\cite{lai14b,gov00}. 
The algorithm is extremely
quick to run, as most of the computational effort goes into the four numerical integrations
of~\eqref{Eq:MeanFieldA}. Note that our method does not involve discretization of the spatial domain,
and numerical integration of~\eqref{Eq:MeanFieldA} can be performed to whatever accuracy is required simply
by decreasing tolerances in the numerical integrator used. (See below for more
discussion of efficiency.)

%For reference, the $2\pi$-periodic solution of
%\[
%   \frac{dy}{dx}=\tau(a+b\cos{x}+c\sin{x}-y)
%\]
%is
%\[
%   y(x)
%= a + \frac{\tau(b\tau - c)}{\tau^{2}+1}\cos x + \frac{\tau(b + c\tau)}{\tau^{2}+1}\sin x.
%\]

Following solutions of the equations that self-consistently describe
a traveling wave as parameters are varied,
we find that they are created in the previously found Hopf bifurcations
of the spatially uniform state $Z_3$, the locations of which are shown in blue in Fig.~\ref{fig:atau}.
Sometimes the bifurcations are supercritical and sometimes subcritical.
This is shown in Fig.~\ref{fig:varya} where we follow a traveling wave as $a$ is varied
for two  different values of $\tau$.
The stability of the traveling waves was checked empirically.
We used the computed triplet $(z(\xi),I(\xi),Y(\xi))$ as an initial condition
in Eqs.~\eqref{Eq:MeanField}, \eqref{eq:dIdt} and~\eqref{eq:dydt}.
By adding a small perturbation to it, we checked
whether the dynamics of~\eqref{Eq:MeanField}, \eqref{eq:dIdt}, \eqref{eq:dydt}
converges to the predicted traveling wave.
If so, we added a symbol (triangle or circle) to Fig.~\ref{fig:varya}.
In this way we found that for $\tau=1.5$, the spatially uniform state
loses stability supercritically with decreasing $a$
and gives rise to a traveling wave of arbitrarily small amplitude
(dashed line in Fig.~\ref{fig:varya}).
But for $\tau=3$ the bifurcation is subcritical
(solid line in Fig.~\ref{fig:varya}), 
giving a region of bistability between the red and blue curves 
(for this value of $\tau$) in Fig.~\ref{fig:atau}
where the spatially uniform state and a traveling wave are both stable.
Saddle-node bifurcations of traveling waves were also found for other values of $\tau$
and are shown by the red curves in Fig.~\ref{fig:atau}.
 
%But can the periodic solution created
%here be followed in DDE-BIFTOOL? Yes, with 32 grid points, but not very far.

%%%%%%%%%
\begin{figure}[htbp]
\centering
\includegraphics[width=0.8\textwidth]{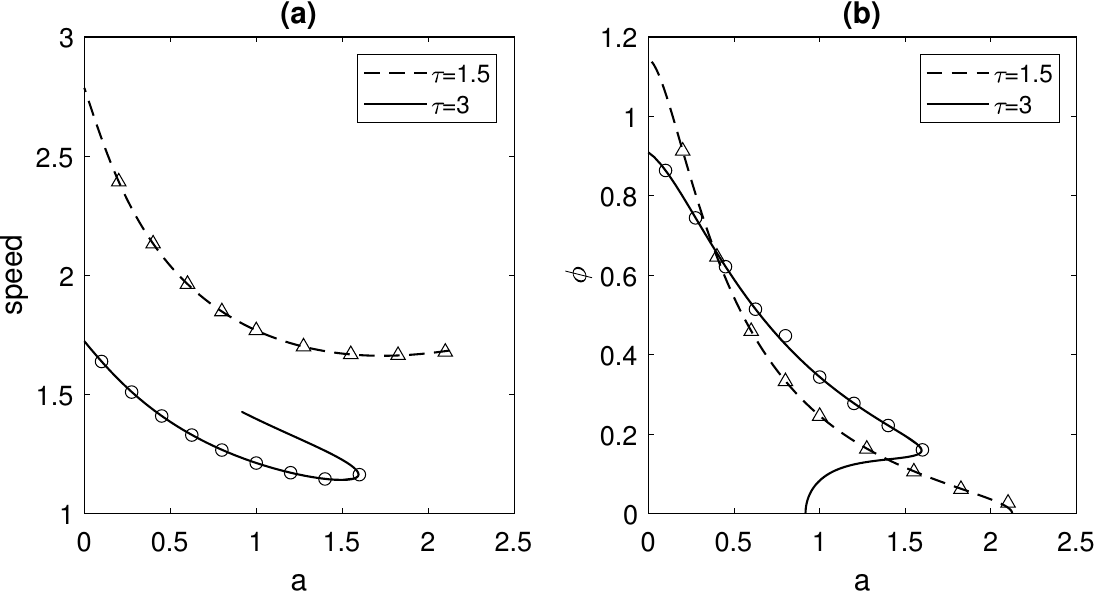}
\caption{(a) Speed of traveling wave solution of~\eqref{Eq:MeanField}--\eqref{eq:Iint} for two
different values of $\tau$ (shown with dashed vertical lines
in Fig.~\ref{fig:atau}). 
(b) $\phi$ (amplitude of pattern) for the solutions shown in panel (a).
The symbols are the results of numerical simulations of~\eqref{Eq:MeanField},~\eqref{eq:dIdt}
and~\eqref{eq:dydt}.
Parameters as in Table~\ref{tab:param}.
}
\label{fig:varya}
\end{figure}

Note that the form of the coupling function $G(x)$, with only the lowest spatial harmonic, 
determines the form of the traveling
waves: a quantity such as firing
frequency can have only one maximum over the domain, at any instant in time. 
A similar approach to that described here has been used to study moving chimeras
in networks of phase oscillators~\cite{ome23}.

\subsection{Periodic solutions with spatial structure}
\label{sec:persln:inf}

We now show that each Hopf bifurcation
in Fig.~\ref{fig:hopfbump} corresponds
to the emergence of a periodic solution
of Eqs.~(\ref{Eq:MeanField}), (\ref{eq:Iint})
from a stationary bump state.
For this, we develop another version
of the self-consistency argument.

Suppose that $z(x,t)$ is a solution of~(\ref{Eq:MeanField}), (\ref{eq:Iint})
such that $z(x,t+T) = z(x,t)$ with some $T>0$.
Then in the case of delay kernel~\eqref{eq:kern},
the triplet $(z(x,t),I(x,t),Y(x,t))$ is a solution
of~\eqref{Eq:MeanField}, \eqref{eq:dIdt}, \eqref{eq:dydt} with the same period~$T$.
Write~\eqref{eq:dydt} as
\begin{eqnarray}
   \pf{Y(x,t)}{t} & = & \frac{1}{a}\left[A\int_{-\pi}^\pi H(z(y,t-\tau))dy
+B\cos{x}\int_{-\pi}^\pi H(z(y,t-\tau))\cos{y}\: dy \right.  \label{eq:dydtA} \\
    & + & \left. B\sin{x}\int_{-\pi}^\pi H(z(y,t-\tau))\sin{y}\: dy-Y(x,t)\right]. \nonumber
\end{eqnarray}
Next, writing $Y(x,t)$ as a Fourier series in space with time-dependent coefficients we see that all
coefficients decay to zero except those multiplying the spatially-constant and $\cos{x}$ and
$\sin{x}$ terms. We shift our solution in space so that the coefficient of the $\sin{x}$ term is 
zero and thus write
\begin{eqnarray}
Y(x,t) &=& a_0 + \sum_{m=1}^M [a_m\cos{(m \gamma t)}+b_m\sin{(m \gamma t)}]  \label{eq:Yser} \\
&+& \left[c_0+\sum_{m=1}^M [c_m\cos{(m \gamma t)}+d_m\sin{(m \gamma t)}]\right]\cos{x} \nonumber
\end{eqnarray}
where $\gamma=2\pi/T$. The input current $I(x,t)$ will have a similar expansion:
\begin{eqnarray*}
   I(x,t) &=& \alpha_0 + \sum_{m=1}^M [\alpha_m\cos{(m \gamma t)}+\beta_m \sin{(m \gamma t)}] \\
&+& \left[\gamma_0+\sum_{m=1}^M [\gamma_m\cos{(m \gamma t)}+\delta_m\sin{(m \gamma t)}]\right]\cos{x}.
\end{eqnarray*}
Note that, strictly speaking, the Fourier expansions of $Y(x,t)$ and $I(x,t)$
contain coefficients with all positive integer indices~$m$,
but since we are only looking for an approximate solution
of~\eqref{Eq:MeanField}, \eqref{eq:dIdt}, \eqref{eq:dydt},
we truncate these series at a finite number~$M$.

Substituting the above expansions into~\eqref{eq:dIdt},
we find $\alpha_0=a_0,\gamma_0=c_0$ and
\begin{align}
   - m \gamma \alpha_m & = ( b_m - \beta_m )/a, \label{eq:alphak} \\[2mm]
   m \gamma \beta_m & = ( a_m - \alpha_m )/a, \\[2mm]
   - m \gamma \gamma_m & = ( d_m - \delta_m )/a, \\[2mm]
   m \gamma \delta_m & = ( c_m - \gamma_m )/a
   \label{eq:deltak}
\end{align}
for $m=1,2,\dots,M$.
These are linear simultaneous equations, so given the coefficients of $Y$,
$\{a_0,\dots, d_M\}$, we can easily find
the coefficients of $I$, $\{\alpha_0,\dots,\delta_M\}$.

Similarly, write the $T$-periodic functions
\be
   A\int_{-\pi}^\pi H\left( z\left(y,t\right) \right)dy=e_0+\sum_{m=1}^M [e_m\cos{(m \gamma t)}+f_m\sin{(m \gamma t)}] \label{eq:Ainta}
\ee
and
\be
  B\int_{-\pi}^\pi H\left( z\left(y,t\right) \right)\cos{y}\: dy= g_0+\sum_{m=1}^M [g_m\cos{(m \gamma t)}+h_m \sin{(m \gamma t)}]. \label{eq:Binta}
\ee
Then the delayed versions of these are
\begin{eqnarray}
 A\int_{-\pi}^\pi H\left( z\left(y,t-\tau\right) \right)dy &=&
e_0 +\sum_{m=1}^M \left\{ [e_m\cos{(m \gamma \tau)}-f_m \sin{(m \gamma \tau)}]\cos{(m \gamma t)} \vphantom{\sum}\right. \label{eq:Aint} \\
&+& \left.\vphantom{\sum} [e_m \sin{(m \gamma \tau)}+f_m \cos{(m \gamma \tau)}]\sin{(m \gamma t)} \right\}
\nonumber
\end{eqnarray}
and
\begin{eqnarray}
B\int_{-\pi}^\pi H\left( z\left(y,t-\tau\right) \right)\cos{y}\: dy &=&
g_0 + \sum_{m=1}^M \left\{ [g_m \cos{(m \gamma \tau)}-h_m \sin{(m \gamma \tau)}]\cos{(m \gamma t)} \vphantom{\sum} \right. \label{eq:Bint} \\
&+& \left.\vphantom{\sum} [g_m \sin{(m \gamma \tau)}+h_m \cos{(m \gamma \tau)}]\sin{(m \gamma t)} \right\}. \nonumber
\end{eqnarray}
Substituting~\eqref{eq:Aint}--\eqref{eq:Bint} and~\eqref{eq:Yser} into~\eqref{eq:dydtA},
we find $a_0=e_0,c_0=g_0$ and
\begin{align}
   - m \gamma a_m & = \left[e_m \sin{(m \gamma \tau)}+f_m\cos{(m \gamma \tau)}-b_m\right]/a, \label{eq:ak} \\[2mm]
   m \gamma b_m & = \left[e_m \cos{(m \gamma \tau)}-f_m \sin{(m \gamma \tau)}-a_m \right]/a, \\[2mm]
   - m \gamma c_m & = \left[g_m \sin{(m \gamma \tau)}+h_m \cos{(m \gamma \tau)}-d_m \right]/a, \\[2mm]
   m \gamma d_m & = \left[g_m \cos{(m \gamma \tau)}-h_m \sin{(m \gamma \tau)}-c_m \right]/a \label{eq:dk}
\end{align}
for $m=1,2,\dots,M$.
These are also simultaneous linear equations so given the coefficients $\{e_0,\dots, f_M\}$ 
and $\{g_0,\dots, h_M\}$ we can easily solve them to find
the coefficients of $Y$, $\{a_0,\dots, d_M\}$.

Thus we have an algorithm for finding approximate periodic solutions 
of~\eqref{Eq:MeanField}, \eqref{eq:dIdt}, \eqref{eq:dydt} in a self-consistent way:
\begin{enumerate}
\item From a simulation of~\eqref{Eq:MeanField}--\eqref{eq:Iint}
for which there is a stable periodic solution, extract one period of $I(x,t)$.
From this, fit the coefficients $\{\alpha_0,\dots,\delta_M\}$. 
Given these coefficients, use them to construct $I(x,t)$
and find the $T$-periodic solution, $z(x,t)$, of~\eqref{Eq:MeanField}
with magnitude $|z(x,t)|<1$. At each value of $x$, $I(x,t)$
is an externally imposed $T$-periodic function,
so the solution at any $x$ can be found
using the same technique as explained in Sec.~\ref{sec:travinf}.
The solutions at different values of $x$ are independent, so these can be found in parallel. 

\item Use this solution to construct $H(z(x,t))$ and evaluate the integrals 
in~\eqref{eq:Ainta}--\eqref{eq:Binta}, and fit the coefficients $\{e_0,\dots, h_M\}$ to these
functions.

\item Given these, solve~\eqref{eq:ak}--\eqref{eq:dk} for $\{a_0,\dots, d_M\}$.

\item Given these, solve~\eqref{eq:alphak}--\eqref{eq:deltak} for $\{\alpha_0,\dots,\delta_M\}$.
We refer to the solutions of these equations as $\{\tilde{\alpha}_0,\dots,\tilde{\delta}_M\}$.
\item If all $\{\tilde{\alpha}_0,\dots,\tilde{\delta}_M\}=\{\alpha_0,\dots,\delta_M\}$
then we are done. There is one more unknown, $T$,
and we append a pinning condition on the solution to find it.
\end{enumerate}

%A typical stable periodic solution is shown in Fig.~\ref{fig:podist}. The firing frequency is
%given by~\eqref{eq:freq}.
%
%%%%%%%%%%
%\begin{figure}[h]
%\includegraphics[width=0.8\textwidth]{podist}
%\caption{A periodic solution of~\eqref{Eq:MeanField}-\eqref{eq:kern}.
%The firing frequency is shown in color.
%Parameters as in Table~\ref{tab:param} and $\tau=1,a=0.6$.
%}
%\label{fig:podist}
%\end{figure}
%%%%%%%%%%%%%%%

Following periodic solutions created in the Hopf bifurcations shown in
Fig.~\ref{fig:hopfbump} for constant $a$ we obtain Fig.~\ref{fig:vtau}.
Each branch of solutions is created at a Hopf bifurcation and ends at the next
Hopf bifurcation. Only one branch has to be found numerically; others
can be plotted using reappearance of periodic solutions in DDEs with fixed
delay~\cite{yanper09}. For this value of $a$ each branch undergoes two saddle-node
bifurcations. Following these as both $\tau$ and $a$ are varied we obtain the
red curves shown in Fig.~\ref{fig:hopfbump}.
Note that the circles in Fig.~\ref{fig:vtau} show the periods
of the solutions of~\eqref{Eq:MeanField}, \eqref{eq:dIdt}, \eqref{eq:dydt}
obtained by direct numerical simulations of this system.
These results demonstrate the good accuracy of the approximation with $M = 5$,
as well as the empirical stability of the solutions found.

%%%%%%%%%
\begin{figure}[htbp]
\centering
\includegraphics[width=0.7\textwidth]{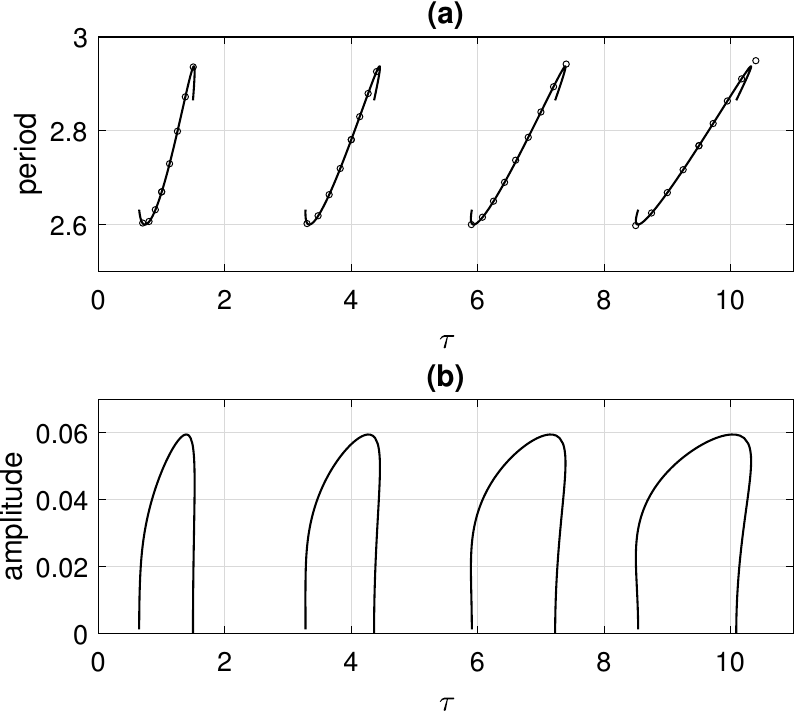}
\caption{Families of periodic solutions created in the Hopf bifurcations
shown in Fig.~\ref{fig:hopfbump} on the black dashed line at $a=0.6$. Panel (a) shows the period while
panel (b) shows the amplitude of oscillations in the spatial mean of $I(x,t)$,
defined as $\left(\sum_{m=1}^M \left[ \alpha_m^2+\beta_m^2 \right] \right)^{1/2}$.
We used $M=5$. The symbols in panel (a) are the results of numerical simulations of~\eqref{Eq:MeanField},~\eqref{eq:dIdt}
and~\eqref{eq:dydt} using 128 spatial points. Parameters as in Table~\ref{tab:param}.
}
\label{fig:vtau}
\end{figure}
%%%%%%%%%%%%%%

Regarding numerical efficiency, following Hopf bifurcations of stationary states is
extremely quick: 500 points on a curve of Hopf bifurcations in Fig.~\ref{fig:atau} 
were found in 0.1s. While following traveling waves,
50 points on a branch of solutions in Fig.~\ref{fig:varya}
were computed in less than 5s, and 50 points on a curve of saddle-node
bifurcations in Fig.~\ref{fig:atau} were found in less that 30s. 
Following periodic solutions and their bifurcations
is slower due to the number of Fourier coefficients used ($4M+2$)
and the need to simulate over the whole spatial domain: 50 points on a branch of periodic
solutions in Fig.~\ref{fig:vtau}
were found in less than 100s, using $M=5$ and 128 spatial points, while 50 points
on a curve of saddle-node bifurcations of periodic solutions
in Fig.~\ref{fig:hopfbump} took around 6 minutes
using the same discretization. These calculations were performed using Matlab on a standard 
desktop computer. Parallel processing was not implemented --- doing so 
significantly speeds up some calculations.

As a comparison we used the DDE-BIFTOOL package~\cite{DDEBIF}, discretizing the spatial
domain into 32 points (quite a coarse discretization), 
giving 128 real-valued DDEs. We did not implement the Jacobian
analytically, instead using finite differences to estimate it. Following a stationary
bump along the black dashed line in Fig.~\ref{fig:hopfbump} was quick, as its profile is
independent of $\tau$. However, finding all of the eigenvalues having
real parts greater than $-0.1$ associated with the bump, for 100 different values of $\tau$,
took around 4 minutes. Once a Hopf bifurcation was found we followed it in the $(\tau,a)$ plane,
tracing out part of a blue curve in Fig.~\ref{fig:hopfbump}. It took around 7 minutes to find
40 points on this curve. DDE-BIFTOOL can follow a branch of periodic solutions created in a
Hopf bifurcation and we did this, using the default approximation of a periodic solution,
to reproduce a curve in Fig.~\ref{fig:vtau}. It took several minutes to find 10 points
on such a curve, and around 10 seconds to find the stability of a single periodic solution
on this curve (a calculation that cannot be done using our method). We did not attempt to
follow saddle-node bifurcations of periodic solutions, given by red curves in
Fig.~\ref{fig:hopfbump}.

Of course, DDE-BIFTOOL is general purpose software and could perform equally well
with any reasonable coupling function $G(x)$. The computational efficiency of our method
relies on (a): the coupling function $G(x)$ having only one harmonic, and (b): the fact that
we are interested in periodic solutions of periodically-forced Riccati equations
and have an efficient way of finding them.

%
%Following the periodic solution shown in Fig.~\ref{fig:podist} as $a$ is varied
%we obtain Fig.~\ref{fig:podistva}.
%
%
%%%%%%%%%%
%\begin{figure}[h]
%\includegraphics[width=0.8\textwidth]{podistva}
%\caption{Period of a periodic solution of~\eqref{Eq:MeanField}-\eqref{eq:kern}.
%Parameters as in Table~\ref{tab:param} and $\tau=4$.
%}
%\label{fig:podistva}
%\end{figure}
%%%%%%%%%%%%%%%
%
%Varying $\tau$ we obtain Fig.~\ref{fig:podistvtau}. There is a loop in the curve
%and it ends by colliding with a solution having $I=a+b\cos{(2\pi t/T)}\cos{x}$
%for some constants $a,b$. 
%
%%%%%%%%%%
%\begin{figure}[h]
%\includegraphics[width=0.8\textwidth]{podistvtau}
%\caption{Period of a periodic solution of~\eqref{Eq:MeanField}-\eqref{eq:kern}.
%Parameters as in Table~\ref{tab:param} and $a=1.2$.
%}
%\label{fig:podistvtau}
%\end{figure}
%%%%%%%%%%%%%%%
%
%Following the saddle-node bifurcation seen in Fig.~\ref{fig:podistva} we obtain Fig.~\ref{fig:snpo}.
%%%%%%%%%%
%\begin{figure}[h]
%\includegraphics[width=0.8\textwidth]{snpo}
%\caption{Curve of saddle-node bifurcations of the periodic solution, followed from the one
%in Fig.~\ref{fig:podistva}.
%Parameters as in Table~\ref{tab:param}.
%}
%\label{fig:snpo}
%\end{figure}
%%%%%%%%%%%%%%%

%There is a stable stationary bump at $\tau=4,a=0.5$. Also at $\tau=0.5,a=1.2$.

\section{Distributed delays --- compact support}
\label{sec:compact}

In this section we consider the neural field equation~(\ref{Eq:MeanField}) with input current containing distributed delays with compact support~(\ref{eq:IintA}).

\subsection{Stability of stationary states}

The stationary states of~\eqref{Eq:MeanField}, \eqref{eq:IintA}
are the same as in Sec.~\ref{sec:old}, 
since their existence is independent of the delay kernel,
provided that it is normalized to unity.
Moreover, the stability analysis of a stationary state $Z(x)$
is performed in the same way as in Sec.~\ref{sec:stabinf}
with the result that now $L(x,\lambda)$ is given by
\begin{eqnarray*}
L(x,\lambda) &=& \bp D'(Z(x)) & 0 \\ 0 & \overline{D'(Z(x))} \ep \bp \lambda-\mu(x) & 0 \\ 0 & \lambda-\overline{\mu(x)} \ep^{-1} \\
&\times& \frac{i\kappa (e^{-\lambda\tau}-e^{-\lambda(\tau+a)})}{\lambda a}
\bp \zeta(x) & \zeta(x) \\ - \overline{\zeta(x)} & -\overline{\zeta(x)} \ep.
\end{eqnarray*}
Following Hopf bifurcations of the spatially uniform state that is stable in
the absence of delays ($Z_3$) we find the blue curves in Fig.~\ref{fig:atauU}.
As in Sec.~\ref{sec:stabinf},
these curves can be mapped to one another
using knowledge of the imaginary part of the eigenvalues.
Similarly, following Hopf bifurcations
of a stable bump we obtain the blue curves in Fig.~\ref{fig:hopfbumpU}.
Note the qualitative similarities with Figs.~\ref{fig:atau} and~\ref{fig:hopfbump}. 
Below we explain how traveling waves
and periodic solutions created in these bifurcations
can be described in a self-consistent way.

%%%%%%%%%
\begin{figure}[htbp]
\centering
\includegraphics[width=0.8\textwidth]{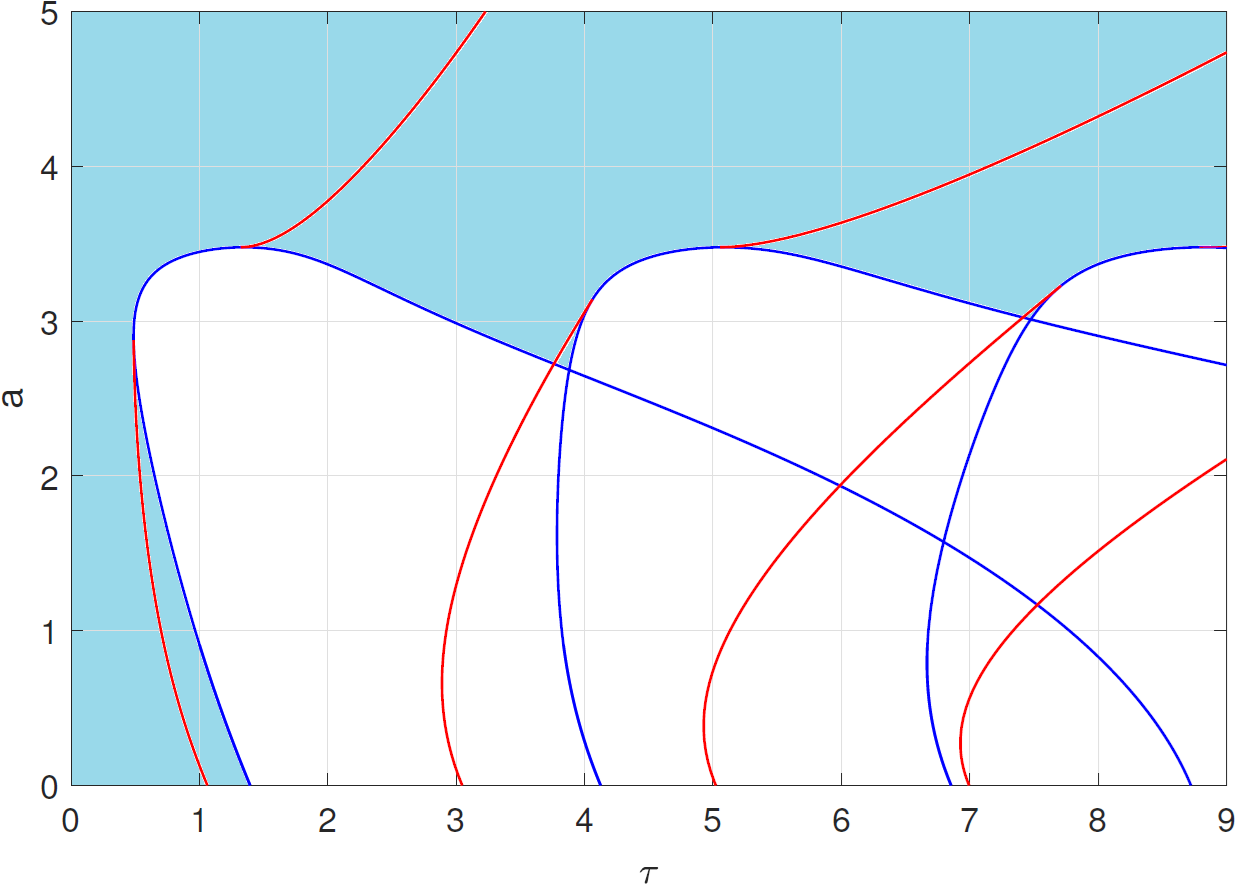}
\caption{Blue: Hopf bifurcations of the spatially uniform steady state $Z_3$
of~\eqref{Eq:MeanField} and~\eqref{eq:IintA}. 
This state is stable in the blue shaded region.
Red: saddle-node bifurcations of
traveling waves.
Parameters as in Table~\ref{tab:param}.}
\label{fig:atauU}
\end{figure}
%%%%%%%%%%%%%%

%%%%%%%%%
\begin{figure}[htbp]
\centering
\includegraphics[width=0.8\textwidth]{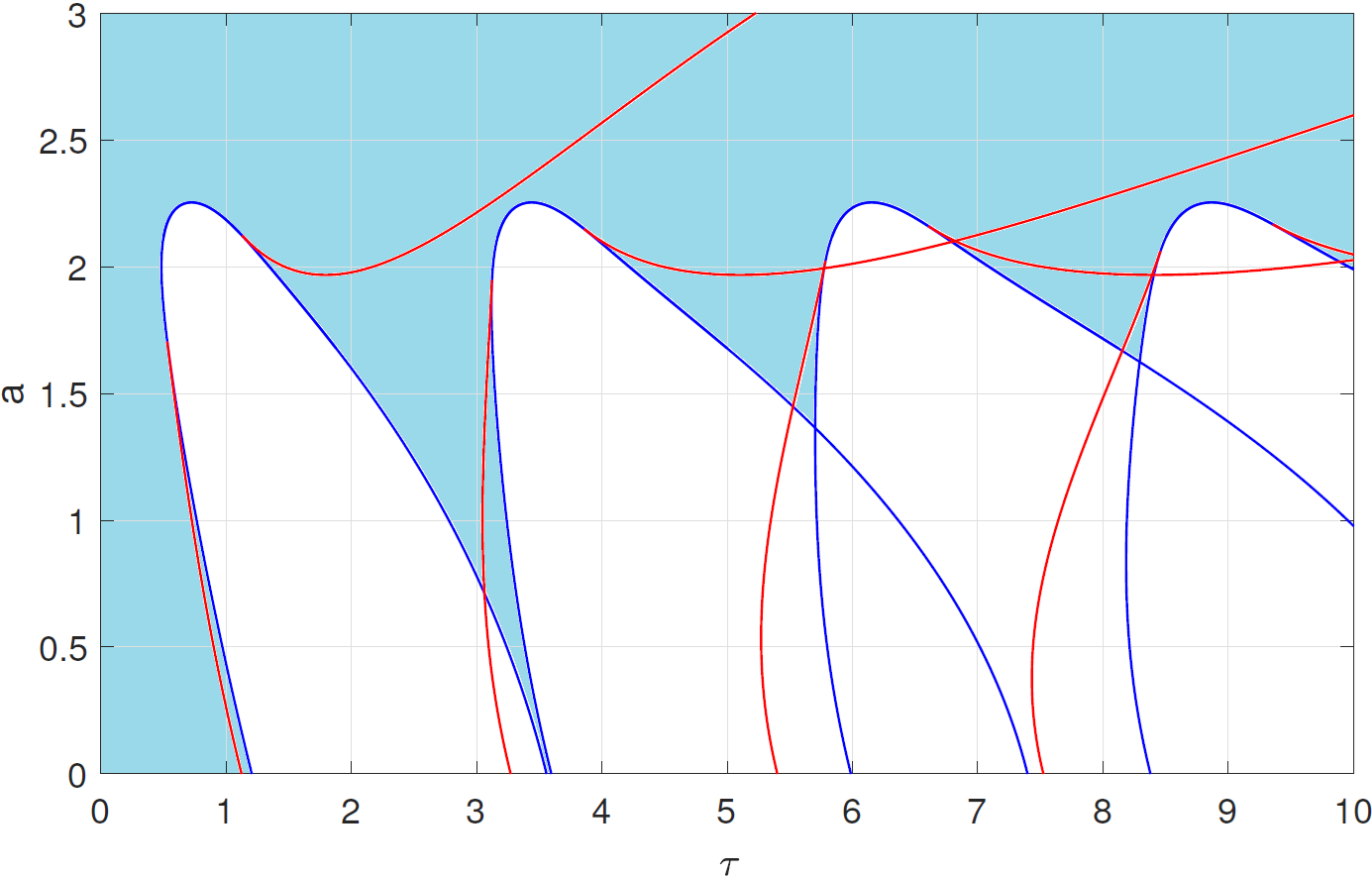}
\caption{Blue: Hopf bifurcations of a stable stationary bump. Red: saddle-node bifurcations of
periodic solutions. Th bump is stable in the blue shaded region.
Parameters as in Table~\ref{tab:param}.}
\label{fig:hopfbumpU}
\end{figure}
%%%%%%%%%%%%%%

\subsection{Traveling waves}
\label{sec:travinf:2}

A traveling wave solution of~\eqref{Eq:MeanField}, \eqref{eq:IintA} 
moving to the left with speed $v$ will have a profile that is a solution of
\begin{equation}
    v\df{z}{\xi} = \fr{[ i ( \eta_0 + \kappa I(\xi)) - \Delta ] ( 1 + z )^2}{2} - i \fr{(1 - z)^2}{2}
\label{Eq:MeanFieldB}
\end{equation}
where
\be
   I(\xi)=\frac{1}{a}\int_{-\pi}^\pi G(\xi-y) \int_{\tau}^{\tau+a} H\left( z\left(y-vs\right) \right)ds\: dy.
\label{eq:Ixi}
\ee
From the form of $G(x)$, see~\eqref{Def:G}, we can write $I(\xi)=\delta+\epsilon\cos{\xi}+\phi\sin{\xi}$
for some constants $\delta,\epsilon$ and $\phi$. Given these constants, we 
can find the $2\pi$-periodic
solution of~\eqref{Eq:MeanFieldB}, $z(\xi)$, with $|z(\xi)|<1$,
as explained above or in~\cite{laiome23,ome23}. 
We then calculate $H(z(y))$. Writing
\[
   H(z(y))=\alpha+\beta\cos{y}+\gamma\sin{y}+\mbox{higher harmonics}
\]
we have
\begin{eqnarray*}
H(z(y-vs)) &=& \alpha+\beta[\cos{y}\cos{(vs)}+\sin{y}\sin{(vs)}] \\[2mm]
&+& \gamma[\sin{y}\cos{(vs)}-\cos{y}\sin{(vs)}]+\mbox{higher harmonics}
\end{eqnarray*}
and thus
\begin{eqnarray*}
   \int_{-\pi}^\pi G(\xi-y)H(z(y-vs))dy &=&
2\pi A\alpha+B\pi \left\{ \cos{\xi}[\beta\cos{(vs)}-\gamma\sin{(vs)}]
\vphantom{\sum}\right. \\
&+& \left.\vphantom{\sum} \sin{\xi}[\beta\sin{(vs)}+\gamma\cos{(vs)}] \right\}.
\end{eqnarray*}
Then, evaluating the integral over $s$ in~\eqref{eq:Ixi} we obtain
\begin{align*}
   I(\xi) & =2\pi A\alpha  + \frac{B\pi}{va}\cos{\xi}\{\beta[\sin{(v(\tau+a))}-\sin{(v\tau)}]+\gamma[\cos{(v(\tau+a))}-\cos{(v\tau)}]\} \nonumber \\
   & + \frac{B\pi}{va}\sin{\xi}\{-\beta[\cos{(v(\tau+a))}-\cos{(v\tau)}]+\gamma[\sin{(v(\tau+a))}-\sin{(v\tau)}]\} \\
   & \equiv \tilde{\delta}+\tilde{\epsilon}\cos{\xi}+\tilde{\phi}\sin{\xi}.
\end{align*}
If these tilded coefficients match the un-tilded ones we are done.
As above we can use the translational invariance of the system to 
set $\epsilon=0$, so in practice we solve 
\[
   \delta-\tilde{\delta}=0 \qquad \mbox{ and } \qquad \phi-\tilde{\phi}=0
\]
along with the pinning condition $\tilde{\epsilon}=0$, 
giving three equations in three unknowns ($\delta,\phi$ and $v$).
Note that for some other delay kernels with compact support, e.g.,~piecewise linear, it is 
also possible to analytically evaluate
the integral over $s$ in~\eqref{eq:Ixi}.

Following traveling waves as $a$ and $\tau$ are varied we find they are destroyed in saddle-node bifurcations,
shown in red in Fig.~\ref{fig:atauU}. Note the qualitative similarity with
Fig.~\ref{fig:atau}.

\subsection{Periodic solutions with spatial structure}

We now show how to describe $T$-periodic solutions of~\eqref{Eq:MeanField}, \eqref{eq:IintA}
in a self-consistent way. 
Due to the form of $G(x)$ and the translational invariance of the system 
we can approximate $I(x,t)$ by the finite Fourier series in time
\begin{eqnarray}
I(x,t) &=& \alpha_0+\sum_{m=1}^M [\alpha_m \cos{(m \gamma t)}+\beta_m \sin{(m \gamma t)}] \label{eq:Iper} \\
&+& \left[\gamma_0+\sum_{m=1}^M [\gamma_m \cos{(m \gamma t)}+\delta_m \sin{(m \gamma t)}]\right]\cos{x} \nonumber
\end{eqnarray}
where $\gamma=2\pi/T$. Let us assume that
we know all of these coefficients for one set of parameter values, as extracted from a simulation. 
We then substitute~\eqref{eq:Iper} into~\eqref{Eq:MeanField} and find
the $T$-periodic solution $z(x,t)$ of~\eqref{Eq:MeanField}.
For this, at each value of $x$, we use the same technique
as explained in Sec.~\ref{sec:travinf}.
Next, we construct $H(z(x,t))$, which will also be even in $x$, and write it as
\begin{eqnarray*}
   H(z(x,t)) &=& a_0+\sum_{m=1}^M [a_m\cos{(m \gamma t)}+b_m \sin{(m \gamma t)}] \\
&+& \left[c_0+\sum_{m=1}^M [c_m \cos{(m \gamma t)}+d_m \sin{(m \gamma t)}]\right]\cos{x}+\mbox{higher spatial harmonics}.
\end{eqnarray*}
Then
\begin{eqnarray*}
   H(z(x,t-s)) &=& a_0+\sum_{m=1}^M \left\{\cos{(m \gamma t)}[a_m \cos{(m \gamma s)}-b_m \sin{(m \gamma s)}] \right. \\
   &+& \left. \sin{(m \gamma t)}[a_m \sin{(m \gamma s)}+b_m \cos{(m \gamma s)}]\right\} \\
   &+& \cos{x}\left(c_0+\sum_{m=1}^M \left\{\cos{(m \gamma t)}[c_m \cos{(m \gamma s)}-d_m \sin{(m \gamma s)}] \right. \right. \\
   &+& \left. \sin{(m \gamma t)}[c_m \sin{(m \gamma s)}+d_m \cos{(m \gamma s)}]\right\}\bigg)+\mbox{higher harmonics}
\end{eqnarray*}
and thus
\begin{eqnarray*}
    \int_{-\pi}^\pi G(x-y)H(z(y,t-s))dy &=& 2\pi A \left[a_0+\sum_{m=1}^M \left\{\cos{(m \gamma t)}[a_m \cos{(m \gamma s)}-b_m \sin{(m \gamma s)}] \right. \right. \\
   &+& \left. \sin{(m \gamma t)}[a_m \sin{(m \gamma s)}+b_m \cos{(m \gamma s)}]\right\}\bigg] \\
   &+& B\pi \cos{x}\left(c_0+\sum_{m=1}^M \left\{\cos{(m \gamma t)}[c_m \cos{(m \gamma s)}-d_m \sin{(m \gamma s)}] \right. \right. \\
   &+& \left. \sin{(m \gamma t)}[c_m \sin{(m \gamma s)}+d_m \cos{(m \gamma s)}]\right\}\bigg).
\end{eqnarray*}
Now we calculate
\begin{eqnarray*}
    I(x,t) &=& \frac{1}{a}\int_{-\pi}^\pi G(x - y) \int_{\tau}^{\tau+a}H\left( z\left(y,t-s\right) \right) ds\: dy = \tilde{\alpha}_0+\sum_{m=1}^M [\tilde{\alpha}_m \cos{(m \gamma t)}+\tilde{\beta}_m \sin{(m \gamma t)}] \\
   &+&
\left[\tilde{\gamma}_0+\sum_{m=1}^M [\tilde{\gamma}_m \cos{(m \gamma t)}+\tilde{\delta}_m \sin{(m \gamma t)}]\right]\cos{x}
\end{eqnarray*}
where $\tilde{\alpha}_0=2\pi Aa_0$ and $\tilde{\gamma}_0=B\pi c_0$ and
\begin{align*}
   \tilde{\alpha}_m & =\frac{2\pi A}{m \gamma a}\{a_m [\sin{(m \gamma (\tau+a))}-\sin{(m \gamma \tau)}]+b_m [\cos{(m \gamma (\tau+a))}-\cos{(m \gamma \tau)}]\}, \\[2mm]
   \tilde{\beta}_m & =\frac{2\pi A}{m \gamma a}\{a_m [-\cos{(m \gamma (\tau+a))}+\cos{(m \gamma \tau)}]+b_m [\sin{(m \gamma (\tau+a))}-\sin{(m \gamma \tau)}]\}, \\[2mm]
    \tilde{\gamma}_m & = \frac{B\pi}{m \gamma a}\{c_m [\sin{(m \gamma (\tau+a))}-\sin{(m \gamma \tau)}]+d_m [\cos{(m \gamma (\tau+a))}-\cos{(m \gamma \tau)}]\}, \\[2mm]
    \tilde{\delta}_m & = \frac{B\pi}{m \gamma a}\{c_m [-\cos{(m \gamma (\tau+a))}+\cos{(m \gamma \tau)}]+d_m [\sin{(m \gamma (\tau+a))}-\sin{(m \gamma \tau)}]\}.
\end{align*}
For self-consistency we need all of the tilded expressions just found 
to equal the non-tilded ones that we started with.
There is one more unknown ($T$) so we add a pinning condition too,
solving the $4M+3$ resulting equations using Newton's method.
Employing the above numerical scheme,
we calculated the arc-shaped branches
of periodic solutions that start and end
at the blue curves in Fig.~\ref{fig:hopfbumpU}.
The resulting diagram looks similar to Fig.~\ref{fig:vtau}, so we do not show it.
Importantly, the found periodic solutions have saddle-node bifurcations, as shown by the red curves in Fig.~\ref{fig:hopfbumpU}.

Note: other stable periodic solutions exist in this model, as seen in Fig.~\ref{fig:other}(a). 
These solutions do not arise directly from Hopf bifurcations of the stable bump. Following this
solution as $\tau$ is varied it undergoes a saddle-node bifurcation as $\tau$ is decreased, before
colliding with another (apparently unstable) periodic solution. Similar solutions were found in the model
studied in Sec.~\ref{sec:inf}.

%%%%%%%%%
\begin{figure}[tbp]
\centering
\includegraphics[width=0.9\textwidth]{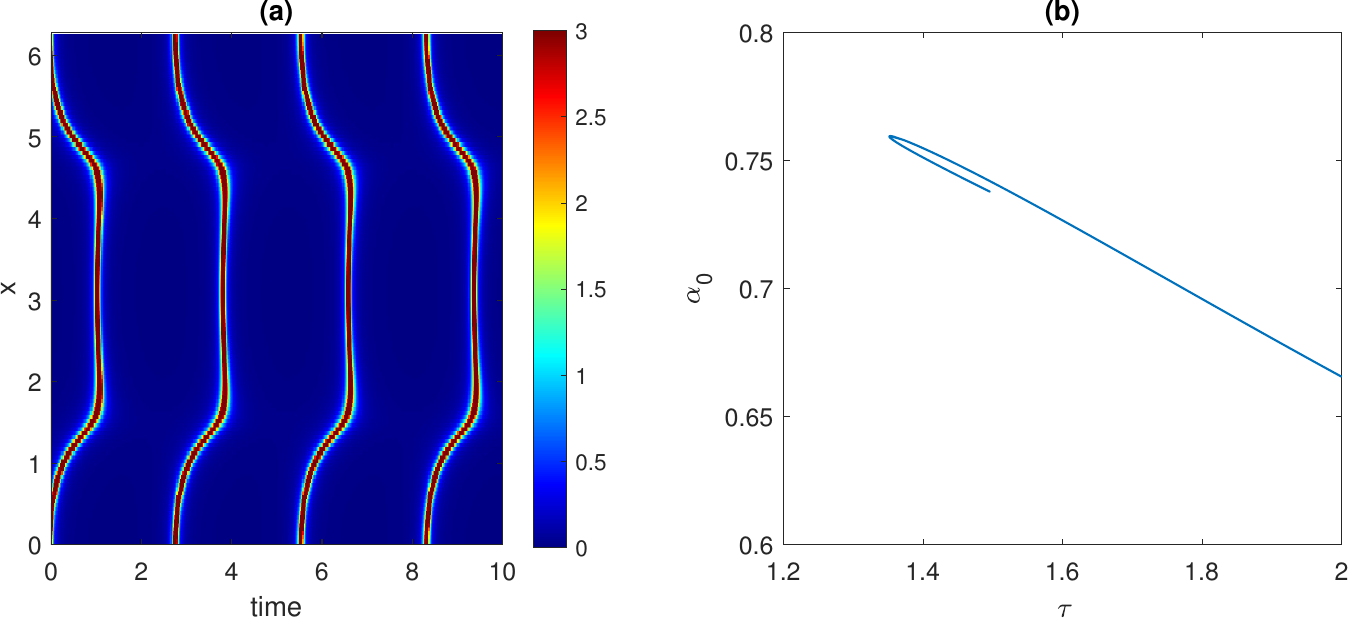}
\caption{(a): a periodic solution of~\eqref{Eq:MeanField} with~\eqref{eq:IintA} for $\tau=1.5$. 
Firing frequency is shown in color, capped at 3.
(b): continuation of the solution in panel (a) as $\tau$ is varied. $a=1$ and other
parameters as in Table~\ref{tab:param}.}
\label{fig:other}
\end{figure}
%%%%%%%%%%%%%%

\section{Conduction delays}
\label{sec:cond}

In this section we consider the neural field equation~\eqref{Eq:MeanField}
with an input current term containing
conduction delays~\eqref{eq:IintC}.
As explained above, this model describes
the long-term behavior of a large population
of theta neurons given by
Eqs.~\eqref{eq:dthetadt} and~\eqref{eq:IdiscB}. 
One of the parameters in this model is conduction speed $c$. This speed is determined
by biophysical properties such as axon diameter and the thickness of its myelin 
sheath~\cite{schkno19}. The properties of myelin may change with both age
and as the result of disorders such as multiple sclerosis~\cite{huazha25}.

\subsection{Stability of stationary states}

Consider a stationary state, $Z(x)$, of~\eqref{Eq:MeanField}, \eqref{eq:IintC}.
Inserting the ansatz $z(x,t) = Z(x) + v(x,t)$ into~\eqref{Eq:MeanField}, \eqref{eq:IintC}
and linearizing with respect to $v(x,t)$, we obtain
\begin{eqnarray}
\pf{v}{t} &=& \mu(x) v + i \kappa \zeta(x)
\int_{-\pi}^\pi G(x - y)  \label{eq:lincond} \\[2mm]
&\times& 
\left[ D'( Z(y) ) v\left(y,t -\tau- \fr{|x - y|}{c} \right)
+ \overline{D'( Z(y) )} \overline{v}\left(y,t - \tau-\fr{|x - y|}{c} \right) \right]\: dy,
\nonumber
\end{eqnarray}
where $\mu(x)$, $\zeta(x)$ and $D'(z)$ are given
by~\eqref{Def:mu}, \eqref{Def:zeta} and~\eqref{Def:D_}, respectively.
Although Eq.~\eqref{eq:lincond} is very similar to Eq.~\eqref{eq:lin},
it has important qualitative differences.
In particular, the dependence on $x - y$ in the integral term
is present not only in the coupling function $G(x)$
but also in the argument of the perturbation $v(x,t)$.
This makes the analysis of Eq.~\eqref{eq:lincond} more complicated.

\subsubsection{Constant stationary states}

Suppose that $Z(x)$ does not depend on~$x$;
then the functions $\mu(x)$ and~$\zeta(x)$ also do not depend on~$x$.
To highlight this feature, we use the simple letters $Z$, $\mu$ and $\zeta$.
The stability of the constant stationary state $Z$ can be investigated using the ansatz
$$
v(x,t) = v_+ e^{i k x} e^{\lambda t} + \overline{v}_- e^{-i k x} e^{\overline{\lambda} t}
$$
with different integers~$k$. Inserting this into Eq.~\eqref{eq:lincond}, we obtain
\begin{eqnarray*}
\lambda v_+ e^{i k x} e^{\lambda t} &+& \overline{\lambda} \overline{v}_- e^{-i k x} e^{\overline{\lambda} t} = \mu v_+ e^{i k x} e^{\lambda t}
+ \mu \overline{v}_- e^{-i k x} e^{\overline{\lambda} t} + i \kappa \zeta
\int_{-\pi}^\pi G(x - y)  \\[2mm]
&\times&
\left[ D'( Z ) \left( v_+ e^{i k y} e^{-\lambda\tau} e^{\lambda t} e^{-\lambda |x - y| / c}
+ \overline{v}_- e^{-i k y} e^{-\overline{\lambda} \tau} e^{\overline{\lambda} t} e^{-\overline{\lambda} |x - y| / c} \right) \right. \\[2mm]
&+&
\left. \overline{D'( Z} \left( \overline{v}_+ e^{-i k y} e^{-\overline{\lambda} \tau} e^{\overline{\lambda} t} e^{-\overline{\lambda} |x - y| / c}
+ v_- e^{i k y} e^{-\lambda \tau} e^{\lambda t} e^{-\lambda |x - y| / c} \right) \right] \: dy.
\end{eqnarray*}
Equating separately the terms proportional to $e^{\lambda t}$ and $e^{\overline{\lambda} t}$,
we obtain a system of two equations
\begin{eqnarray*}
\lambda v_+ e^{i k x} &=& \mu v_+ e^{i k x} + i \kappa \zeta
\int_{-\pi}^\pi G(x - y)  \\[2mm]
&\times&
\left[ D'( Z ) v_+ e^{i k y} e^{-\lambda\tau} e^{-\lambda |x - y| / c}
+ \overline{D'( Z )} v_- e^{i k y} e^{-\lambda\tau} e^{-\lambda |x - y| / c} \right] \: dy
\end{eqnarray*}
and
\begin{eqnarray*}
\lambda v_- e^{i k x} &=&
\overline{\mu} v_- e^{i k x} - i \kappa \overline{\zeta}
\int_{-\pi}^\pi G(x - y)  \\[2mm]
&\times&
\left[ \overline{D'( Z )} v_- e^{i k y} e^{-\lambda\tau} e^{-\lambda |x - y| / c}
+ D'( Z ) v_+ e^{i k y} e^{-\lambda\tau} e^{-\lambda |x - y| / c} \right] \: dy.
\end{eqnarray*}
Next, for every $2\pi$-periodic function $f(x)$ we have
$$
f(x) = \sum\limits_{k=-\infty}^\infty \hat{f}_k e^{i k x}
\quad\mbox{where}\quad
\hat{f}_k = \fr{1}{2\pi} \int_{-\pi}^\pi f(x) e^{-i k x}.
$$
Therefore,
\[
\lambda v_+ = \mu v_+ + \fr{i \kappa \zeta}{2\pi} \left[ D'( Z ) v_+ + \overline{D'( Z )} v_- \right] e^{-\lambda\tau} \int_{-\pi}^\pi \int_{-\pi}^\pi e^{- i k (x - y)} G(x - y)  e^{-\lambda |x - y| / c} \: dy\: dx
\]
and
\[
\lambda v_- = \overline{\mu} v_- - \fr{i \kappa \overline{\zeta}}{2\pi} \left[ D'( Z ) v_+ + \overline{D'( Z )} v_- \right] 
e^{-\lambda\tau} \int_{-\pi}^\pi \int_{-\pi}^\pi e^{- i k (x - y)} G(x - y)  e^{-\lambda |x - y| / c} \: dy\: dx.
\]
Thus
\be
\lambda \left(
\begin{array}{c}
v_+\\[2mm]
v_-
\end{array}
\right)
=
\left(
\begin{array}{cc}
\mu & 0\\[2mm]
0 & \overline{\mu}
\end{array}
\right)
\left(
\begin{array}{c}
v_+\\[2mm]
v_-
\end{array}
\right)
+ i \kappa Q_k(\lambda)e^{-\lambda\tau}
\left(
\begin{array}{cc}
\zeta & \zeta \\[2mm]
- \overline{\zeta} & - \overline{\zeta}
\end{array}
\right)
\left(
\begin{array}{c}
D'( Z ) v_+\\[2mm]
\overline{D'( Z )} v_-
\end{array}
\right) \label{eq:stab}
\ee
where
\[
Q_k(\lambda) = \fr{1}{2\pi} \int_{-\pi}^\pi \int_{-\pi}^\pi e^{- i k (x - y)} G(x - y)  e^{-\lambda |x - y| / c} \: dx\: dy
=  \int_{-\pi}^\pi G(x) e^{- i k x} e^{-\lambda |x| / c} dx .
\]
In particular, for the function $G(x)$ given by~\eqref{Def:G}, we have
\begin{eqnarray*}
Q_k(\lambda)
%&=& \int_{-\pi}^0 (A+B\cos{x})e^{- i k x} e^{\lambda x / c} dx
%+\int_{0}^\pi (A+B\cos{x})e^{- i k x} e^{-\lambda x / c} dx \\
&=& \frac{2A(\lambda/c)}{(\lambda/c)^2+k^2}\Big(1-(-1)^k e^{-\lambda\pi/c}\Big) \\
&+& B\Big(1+(-1)^k e^{-\lambda\pi/c}\Big)
\left[\frac{\lambda/c}{(\lambda/c)^2+(k+1)^2}+\frac{\lambda/c}{(\lambda/c)^2+(k-1)^2}\right].
\end{eqnarray*}
For each $k\in\mathbb{Z}$, for~\eqref{eq:stab} to have non-zero solutions
for $v_+$ and $v_-$ we require
\begin{equation}
|\lambda I_2-K_k(\lambda)|=0
\label{Eq:K_k}
\end{equation}
where
\[
   K_k(\lambda)=\left(
\begin{array}{cc}
\mu & 0\\[2mm]
0 & \overline{\mu}
\end{array}
\right)
+ i \kappa Q_k(\lambda)e^{-\lambda\tau}
\left(
\begin{array}{cc}
\zeta & \zeta \\[2mm]
- \overline{\zeta} & - \overline{\zeta}
\end{array}
\right)
\begin{pmatrix} D'( Z ) & 0 \\ 0 & \overline{D'( Z )}\end{pmatrix}.
\]
Clearly, eigenvalues $\lambda$ associated with the stability of the stationary state $Z$
are given by solutions of Eq.~\eqref{Eq:K_k}.
Eigenvalues with zero real part correspond to instabilities of the stationary state.
For an illustrative example, we choose the parameter values
given in Table~\ref{tab:param:2} and $\eta_0 = 0.1$.
Instabilities of the spatially uniform state that is stable in the absence of delays
(analogous to $Z_3$ in the previous section),
where eigenvalues have non-zero imaginary parts 
for $k=0,1,2,3,4,5$, are shown in Fig.~\ref{fig:conduni}.
The curves for $k=2,3,4,5$ have similar shapes,
while those for $k=0,1$ are qualitatively different.
This spatially uniform state is stable when $\tau=1/c=0$
and can lose stability to perturbations with $k=0,1,2,3,4$,
depending on how $\tau$ and $c$ are varied away from this point in parameter space,
and which curve of instabilities is crossed first.

\begin{table}[htbp]
\footnotesize
\caption{Parameter values}
\begin{center}
\begin{tabular}{|c||ccccccccc|} \hline
Parameter\vphantom{$\int_0^1$} & & $\Delta$ & & $\kappa$ & & $A$ & & $B$ & \\ \hline
Value\vphantom{$\int_0^1$} & & 0.03 & & 1.5 & & 0.1 & & 0.3 & \\ \hline
\end{tabular}
\end{center}
\label{tab:param:2}
\end{table}

\begin{figure}[htbp]
\centering
\includegraphics[width=0.8\textwidth]{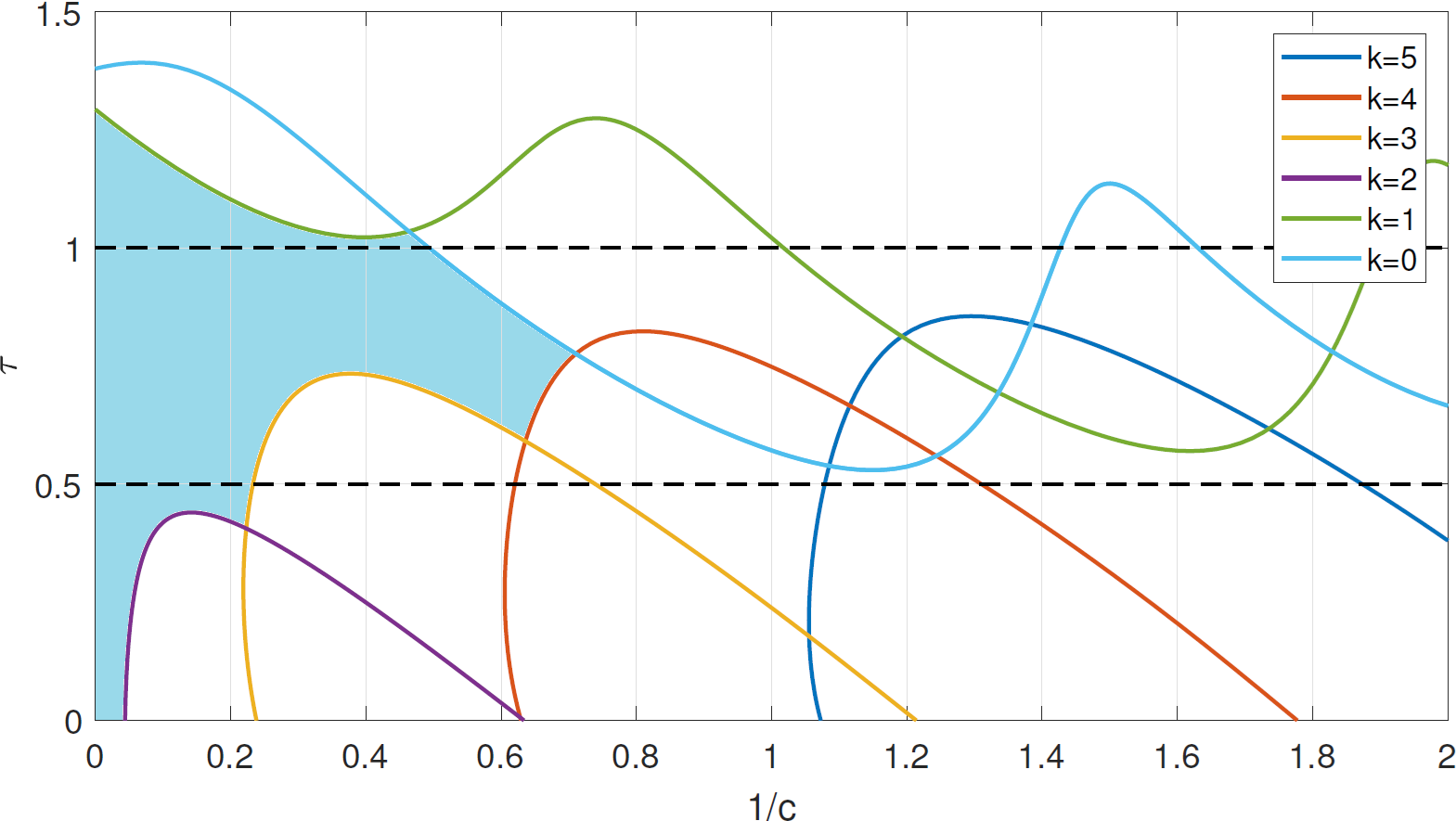}
\caption{Hopf instabilities of a spatially uniform state 
of Eq.~(\ref{Eq:MeanField}), (\ref{eq:IintC}), for perturbations with different wavenumbers, $k$. The state is stable when $\tau=1/c=0$ and in the blue shaded region.
The black dashed lines relate to later Figures.
%The state becomes unstable as $\tau$ is decreased. 
Parameters: $\eta_0=0.1$ and other parameters as in Table~\ref{tab:param:2}.}
\label{fig:conduni}
\end{figure}

\subsubsection{Nonconstant stationary states}
\label{sec:cond:bump}

For a nonconstant solution $z = Z(x)$ of~(\ref{Eq:MeanField}), (\ref{eq:IintC}),
we need to use a general perturbation ansatz
$$
v(x,t) = v_+(x) e^{\lambda t} + \overline{v}_-(x) e^{\overline{\lambda} t}.
$$
Inserting this into Eq.~(\ref{eq:lincond}), we obtain
\begin{eqnarray*}
\lambda v_+(x) e^{\lambda t} &+& \overline{\lambda} \overline{v}_-(x) e^{\overline{\lambda} t}
= \mu(x) v_+(x) e^{\lambda t} + \mu(x)  \overline{v}_-(x) e^{\overline{\lambda} t}
+ i \kappa \zeta(x)
\int_{-\pi}^\pi G(x - y) \\[2mm]
&\times&
\left[ D'( Z(y) ) v_+(y) e^{\lambda t} e^{-\lambda \tau} e^{-\lambda |x - y|/c}
+ D'( Z(y) ) \overline{v}_-(y) e^{\overline{\lambda} t}
e^{-\overline{\lambda} \tau} e^{- \overline{\lambda} |x - y|/c}
\right.\\[2mm]
&+&
\left.
\overline{D'( Z(y) )} \overline{v}_+(y) e^{\overline{\lambda} t} e^{-\overline{\lambda} \tau} e^{-\overline{\lambda} |x - y|/c}
+ \overline{D'( Z(y) )} v_-(y) e^{\lambda t}
e^{-\lambda \tau} e^{- \lambda |x - y|/c} \right] dy.
\end{eqnarray*}
Then, separately equating the terms proportional to $e^{\lambda t}$
and the terms proportional to $e^{\overline{\lambda} t}$,
we get two equations that can be written in the matrix form
\begin{eqnarray}
\lambda \left(
\begin{array}{c}
v_+\\[2mm]
v_-
\end{array}
\right)
&=&
\left(
\begin{array}{cc}
\mu(x) & 0\\[2mm]
0 & \overline{\mu}(x)
\end{array}
\right)
\left(
\begin{array}{c}
v_+\\[2mm]
v_-
\end{array}
\right)
+ i \kappa e^{-\lambda \tau}
\left(
\begin{array}{cc}
\zeta(x) & \zeta(x) \\[2mm]
- \overline{\zeta(x)} & - \overline{\zeta(x)}
\end{array}
\right) \label{EVP:Nonconstant}\\[2mm]
&\times&
\int_{-\pi}^\pi G(x - y) e^{- \lambda |x - y|/c}
\left(
\begin{array}{c}
D'( Z(y) ) v_+(y)\\[2mm]
\overline{D'( Z(y) )} v_-(y)
\end{array}
\right) dy.
\nonumber
\end{eqnarray}
As in previous sections, the 
nonlinear eigenvalue problem~(\ref{EVP:Nonconstant}) defines two types of spectra. The essential spectrum is given explicitly
$$
\sigma_\mathrm{ess} = \{ \mu(x)\::\: x\in[-\pi,\pi] \} \cup \{ \mathrm{c.c} \}.
$$
The remaining discrete spectrum $\sigma_\mathrm{discr}$ consists
of isolated eigenvalues $\lambda\notin\sigma_\mathrm{ess}$,
for which Eq.~(\ref{EVP:Nonconstant})
has nontrivial bounded solutions.
This part of the spectrum can be found
only approximately by assuming
$$
\left(
\begin{array}{c}
v_+(x)\\[2mm]
v_-(x)
\end{array}
\right) \approx \sum\limits_{k=-M}^M \left(
\begin{array}{c}
v_{k,+}\\[2mm]
v_{k,-}
\end{array}
\right) e^{i k x}
$$
with some positive integer $M$,
inserting this into the equation
\begin{eqnarray*}
\left(
\begin{array}{c}
v_+(x)\\[2mm]
v_-(x)
\end{array}
\right)
&=& i \kappa e^{-\lambda \tau}
\left(
\begin{array}{cc}
\lambda - \mu(x) & 0\\[2mm]
0 & \lambda - \overline{\mu}(x)
\end{array}
\right)^{-1}
\left(
\begin{array}{cc}
\zeta(x) & \zeta(x) \\[2mm]
- \overline{\zeta(x)} & - \overline{\zeta(x)}
\end{array}
\right)\\[2mm]
&\times&
\int_{-\pi}^\pi G(x - y) e^{- \lambda |x - y|/c}
\left(
\begin{array}{c}
D'( Z(y) ) v_+(y)\\[2mm]
\overline{D'( Z(y) )} v_-(y)
\end{array}
\right) dy,
\end{eqnarray*}
and then projecting the resulting equation
onto the subspace spanned
by the Fourier modes $e^{i k x}$ with $|k|\le M$.
As a result, one obtains a nonlinear eigenvalue
problem in the $2(2M+1)$-dimensional space. Essentially one forms a
$2(2M+1)\times 2(2M+1)$ matrix that depends on $\lambda$. Values of $\lambda$ for
which the determinant of this matrix is zero are eigenvalues associated with the
stability of the solution $Z(x)$. Using $M=10$ we obtain the results shown
in Fig.~\ref{fig:percond} for the stability of a stationary bump, 
for two different values of $\eta_0$.

\subsection{Traveling waves}

Instabilities of the constant stationary state with wavenumbers $k\ne 0$
result in the creation of traveling waves.
The corresponding solution branches can be calculated
using the self-consistency approach described below.
Supposing that Eq.~(\ref{Eq:MeanField}), (\ref{eq:IintC}) has a solution of the form $z = u(x + v t)$,
then $u$ satisfies
\begin{equation}
v \df{u}{x} = \fr{[ i ( \eta_0 + \kappa I(x) ) - \Delta ] ( 1 + u )^2}{2} - i \fr{(1 - u)^2}{2}
\label{Eq:a}
\end{equation}
where
\begin{equation}
I(x) = \int_{-\pi}^\pi G(\xi)
H\left( u\left(x-v\tau- \xi - \fr{v |\xi|}{c} \right) \right)  d\xi.
\label{Def:I:cond_delay}
\end{equation}
Note that in this case, the input current $I(x)$
cannot be represented as a linear combination of a finite set of functions,
as was the case in Sec.~\ref{sec:travinf} and Sec.~\ref{sec:travinf:2}.

\begin{proposition}
Suppose we have
\begin{equation}
   H(u(x)) = a_0 + \sum_{m=1}^\infty \left( a_m\cos{(mx)}+b_m\sin{(mx)} \right).
\label{Def:H:u}
\end{equation}
Then for the coupling function $G(x)$ given by~\eqref{Def:G} we have
\begin{eqnarray*}
  I(x) & = & 2a_0 A\pi+\sum_{m=1}^\infty a_m\left[\cos{(m(x-v\tau))}(C_m^-+C_m^+)+\sin{(m(x-v\tau))}(D_m^-+D_m^+)\right] \\
  & & +\sum_{m=1}^\infty b_m\left[\sin{(m(x-v\tau))}(C_m^-+C_m^+)-\cos{(m(x-v\tau))}(D_m^-+D_m^+)\right],
\end{eqnarray*}
where
\begin{eqnarray*}
C_m^\pm & = & \frac{Ac}{m(c\pm v)}\sin{\left(\frac{m(c\pm v)\pi}{c}\right)}
+ \frac{Bc}{2m}\left[\frac{1}{c\pm v-c/m}\sin{\left(\frac{m(c\pm v-c/m)\pi}{c}\right)} \right. \\
&+& \left. \frac{1}{c\pm v+c/m}\sin{\left(\frac{m(c\pm v+c/m)\pi}{c}\right)}\right], \\[2mm]
D_m^\pm &=& \pm \frac{Ac}{m(c\pm v)}\left[1-\cos{\left(\frac{m(c\pm v)\pi}{c}\right)}\right] 
\pm \frac{Bc}{2m}\left[\frac{1}{c\pm v-c/m}\left\{1-\cos{\left(\frac{m(c\pm v-c/m)\pi}{c}\right)}\right\} \right. \\
&+& \left. \frac{1}{c\pm v+c/m}\left\{1-\cos{\left(\frac{m(c\pm v+c/m)\pi}{c}\right)}\right\}\right].
\end{eqnarray*}
\label{Prop:HtoI}
\end{proposition}

\begin{proof}
Simple calculations with formula~\eqref{Def:H:u} yield
\begin{eqnarray*}
   & & H\left( u\left( x - v\tau - \left( 1 \pm \frac{v}{c} \right) \xi \right) \right) \\
   & = & a_0+\sum_{m=1}^\infty a_m\left[\cos{(m(x-v\tau))}\cos{\left(m\left( 1 \pm \frac{v}{c} \right)\xi\right)}+\sin{(m(x-v\tau))}\sin{\left(m\left( 1 \pm \frac{v}{c} \right)\xi\right)}\right] \\
   &  & +\sum_{m=1}^\infty b_m\left[\sin{(m(x-v\tau))}\cos{\left(m\left( 1 \pm \frac{v}{c} \right)\xi\right)}-\cos{(m(x-v\tau))}\sin{\left(m\left( 1 \pm \frac{v}{c} \right)\xi\right)}\right].
\end{eqnarray*}
Then using the identity
\begin{eqnarray*}
\int_{-\pi}^\pi G(\xi) H\left( u\left( x - v\tau - \xi - \fr{v |\xi|}{c} \right) \right) d\xi
&=& \int_{-\pi}^0 G(\xi) H\left( u\left( x - v\tau - \left( 1 - \frac{v}{c} \right) \xi \right) \right) d\xi \\
&+& \int_0^{\pi} G(\xi) H\left( u\left( x - v\tau - \left( 1 + \frac{v}{c} \right) \xi \right) \right) d\xi
\end{eqnarray*}
and evaluating the integrals on the right-hand side with $G(x) = A + B\cos{(x)}$,
we obtain the stated Fourier expansion for $I(x)$.
\end{proof}

%A stable 
%3-twisted wave is shown in Fig.~\ref{fig:waveB}.
%
%\begin{figure}
%\begin{center}
%\includegraphics[width=12cm]{waveB}
%\caption{$H$ is shown in colour. Parameters:
%$A=0.1,B=0.3,\gamma=0.03,\eta_0=0.1,\kappa=1.5,c=1.5,\tau=0.5$}
%\label{fig:waveB}
%\end{center}
%\end{figure}

An approximate solution of \eqref{Eq:a}--\eqref{Def:I:cond_delay}
can be found using the following self-consistency approach.
Fix the number of terms in~\eqref{Def:H:u}, truncating the series at $m=M$.
Given $a_0,\dots,b_M$, use Proposition~\ref{Prop:HtoI} to construct
the approximate Fourier expansion of $I(x)$ with the same number of terms $M$.
Next, use the same technique as explained in Sec.~\ref{sec:travinf}
to find the periodic solution $u(x)$ such that $|u(x)| < 1$.
This is the profile of the traveling wave.
Construct $H(u(x))$. Its Fourier coefficients have to equal $a_0,\dots, b_M$.
(A pinning condition is added too, in order to find the wave speed $v$.)

Following waves with different wave numbers as $c$ is varied for a fixed delay $\tau=0.5$
we obtain Fig.~\ref{fig:trav}. Referring back to Fig.~\ref{fig:conduni},
we see that as $1/c$ is increased each wave is created in a supercritical bifurcation
from the spatially uniform state and then undergoes a saddle-node bifurcation,
with the branch of solutions ending in a subcritical bifurcation
from the spatially uniform state.
The circles in Fig.~\ref{fig:trav} show the results from numerical simulations
of Eq.~\eqref{Eq:MeanField}, \eqref{eq:IintC}.
As expected, only the branch of traveling waves with $k=3$ is stable,
while the other branches with $k=4$ and $k=5$ are unstable,
as they bifurcate from the already unstable constant stationary state.

\begin{figure}[htbp]
\centering
\includegraphics[width=0.7\textwidth]{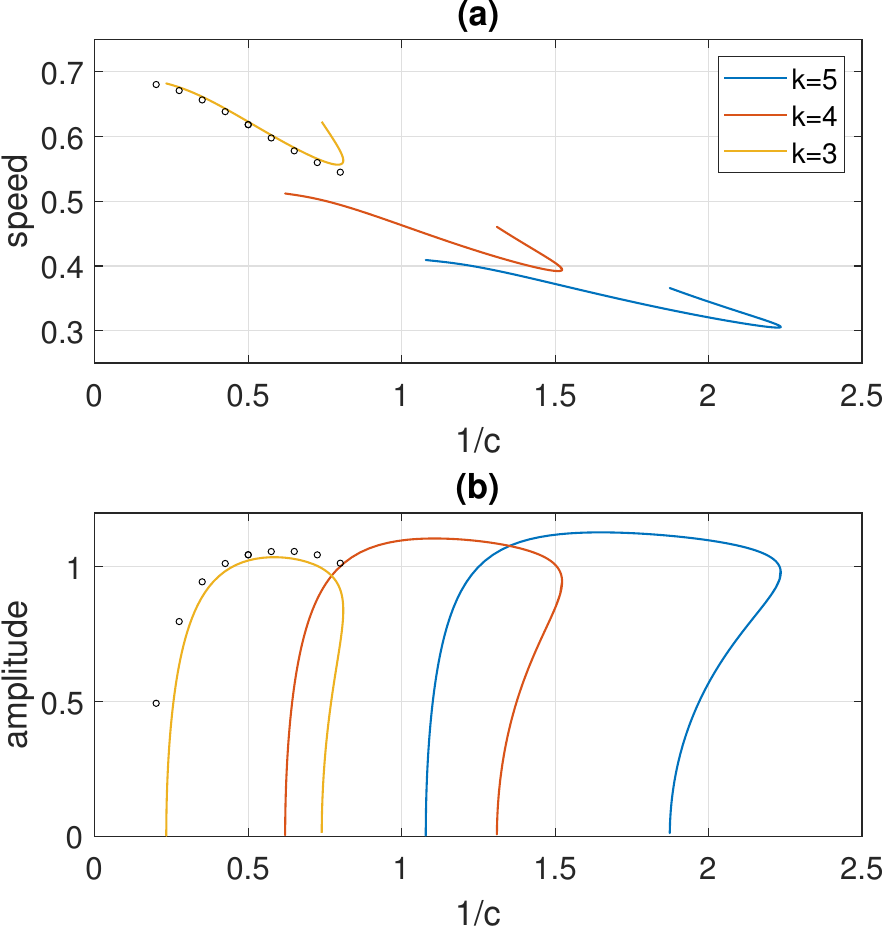}
\caption{Traveling wave solutions of Eq.~(\ref{Eq:MeanField}), (\ref{eq:IintC}).
$k$ is the angular wave number ($2\pi$ divided by the wavelength).  
(a): wave speed $v$. (b): wave amplitude, defined as $\left(\sum_{m=1}^M \left[ a_m^2+b_m^2 \right] \right)^{1/2}$. 
The symbols are from numerical simulations 
of Eq.~(\ref{Eq:MeanField}), (\ref{eq:IintC}).
The waves with $k=4,5$ are numerically unstable.
Parameters: $\tau=0.5$, $\eta_0=0.1$, $M=20$,
and other parameters as in Table~\ref{tab:param:2}.}
\label{fig:trav}
\end{figure}

\subsection{Periodic solutions with no spatial structure}

The instability of the constant stationary state
with wavenumber $k=0$ is a special case.
It corresponds to a Hopf bifurcation,
which gives rise to solutions of Eq.~\eqref{Eq:MeanField}, \eqref{eq:IintC}
that are periodic in time but have no spatial structure,
i.e.~these solutions are the same at every point in space.
The curve of such instabilities is shown light-blue in Fig.~\ref{fig:conduni}.

%\begin{figure}
%\begin{center}
%\includegraphics[width=12cm]{k=0}
%\caption{Curve of Hopf instability of spatially uniform state 
%of Eq.~(\ref{Eq:MeanField}), (\ref{eq:IintC}), for perturbations with wavenumber $k=0$.  Parameters as in Fig.~\ref{fig:conduni}.}
%\label{fig:k=0}
%\end{center}
%\end{figure}

In this section, we will show how periodic solutions with no spatial structure
can be described using the self-consistency argument.
Suppose that $Z(t)$ satisfies Eq.~\eqref{Eq:MeanField}, \eqref{eq:IintC}
and $Z(t+T) = Z(t)$ for some $T > 0$. Then we have
$$
W(t) \equiv H(Z(t)) = w_{0} + \sum\limits_{m=1}^\infty \left[ w_{m}^\mathrm{c} \cos(m \gamma t) + w_{m}^\mathrm{s} \sin(m \gamma t) \vphantom{\sum} \right],
$$
where $\gamma=2\pi/T$, and therefore
\begin{eqnarray*}
W\left(t - \tau-\fr{|\xi|}{c} \right) &=&  w_{0} \vphantom{\sum\limits_{m=1}^\infty}
+ \sum\limits_{m=1}^\infty \left[ w_{m}^\mathrm{c} \cos( m \gamma(t-\tau)) \cos\left( \fr{m \gamma |\xi|}{c} \right) 
+ w_{m}^\mathrm{c} \sin( m \gamma (t-\tau)) \sin\left( \fr{m \gamma |\xi|}{c} \right) \right. \\[2mm]
&+&
\left. w_{m}^\mathrm{s} \sin( m \gamma(t-\tau)) \cos\left( \fr{m \gamma |\xi|}{c} \right) 
- w_{m}^\mathrm{s} \cos( m \gamma(t-\tau)) \sin\left( \fr{m \gamma |\xi|}{c} \right) \right] .
\end{eqnarray*}
Using~\eqref{eq:IintC} and~\eqref{Def:G}, we write
\begin{eqnarray*}
I(t) &=&  w_{0} I_0
+ \sum\limits_{m=1}^\infty \left[ w_{m}^\mathrm{c} \cos( m \gamma(t-\tau)) I_{m}^\mathrm{c} 
+ w_{m}^\mathrm{c} \sin( m \gamma(t-\tau)) I_{m}^\mathrm{s}
\vphantom{\sum\limits_{m=1}^\infty}  \right.  \nonumber\\[2mm]
&+&
\left. \vphantom{\sum\limits_{m=1}^\infty}
w_{m}^\mathrm{s} \sin( m \gamma(t-\tau)) I_{m}^\mathrm{c}
- w_{m}^\mathrm{s} \cos( m \gamma (t-\tau)) I_{m}^\mathrm{s} \right],
\end{eqnarray*}
where
\begin{eqnarray*}
I_0 &=& \int_{-\pi}^{\pi} G(\xi) d\xi = 2\pi A, \\[2mm]
I_{m}^\mathrm{c} &=& \int_{-\pi}^\pi G(\xi)
\cos\left(\fr{m \gamma |\xi|}{c}\right) d\xi = 2\int_{0}^\pi G(\xi)
\cos\left(\fr{m \gamma \xi}{c}\right) d\xi
= \fr{2 A c}{m \gamma} \sin\left( \fr{m \gamma}{c} \pi \right) \\
&+& B \left[ \fr{c}{m \gamma - c} \sin\left( \fr{m \gamma - c}{c} \pi \right) + \fr{c}{m \gamma + c} \sin\left( \fr{m \gamma+ c}{c} \pi \right) \right], \\[2mm]
I_{m}^\mathrm{s} &=& \int_{-\pi}^\pi G(\xi) 
\sin\left(\fr{m \gamma |\xi|}{c}\right) d\xi = 2\int_{0}^\pi G(\xi) 
\sin\left(\fr{m \gamma \xi}{c}\right) d\xi
= \fr{2 A c}{m \gamma} \left( 1 - \cos\left( \fr{m \gamma}{c} \pi \right) \right) \\
&+& B \left[ \fr{c}{m \gamma - c}
\left( 1 - \cos\left( \fr{m \gamma - c}{c} \pi \right) \right)+ \fr{c}{m \gamma + c} \left( 1 - \cos\left( \fr{m \gamma + c}{c} \pi \right) \right) \right].
\end{eqnarray*}
Given~$I(t)$, we can find the $T$-periodic solution~$\tilde{Z}(t)$
of Eq.~\eqref{Eq:MeanField} using the technique explained in Sec.~\ref{sec:travinf}.
Then we form $\tilde{W}(t)=H(\tilde{Z}(t))$, which has expansion
$$
\tilde{W}(t) = \tilde{w}_{0} + \sum\limits_{m=1}^\infty \left[ \tilde{w}_{m}^\mathrm{c} \cos(m \gamma t) + \tilde{w}_{m}^\mathrm{s} \sin(m \gamma t) \vphantom{\sum} \right].
$$
Finally, we need all of the tilded coefficients to equal their un-tilded versions
(and a pinning condition is added, in order to find the period $T$).
Similar to the previous sections, we work with truncated Fourier expansions
containing only coefficients with indices $m=0,1,\dots,M$.
This allows us to find an approximate periodic solution.

Following such a solution for $\tau = 1$, $\eta_0=0.1$
and other parameters as in Table~\ref{tab:param:2},
we obtain Fig.~\ref{fig:spatuni}.
Referring back to Fig.~\ref{fig:conduni},
we see that the branch of solutions is created in Hopf bifurcations
of the spatially uniform state corresponding to perturbations with wave number $k=0$.
The left endpoint of the branch in Fig.~\ref{fig:spatuni}
is a subcritical bifurcation, so the periodic solutions with small amplitudes are unstable.
Moreover, results of numerical simulations for Eq.~\eqref{Eq:MeanField}, \eqref{eq:IintC}
(circles) indicate that these solutions are stable
from the fold bifurcation up to a value of $1/c$ between $1.3$ and $1.35$,
which is located at a certain distance
from the right endpoint of the branch in Fig.~\ref{fig:spatuni}.

\begin{figure}[htbp]
\centering
\includegraphics[width=0.8\textwidth]{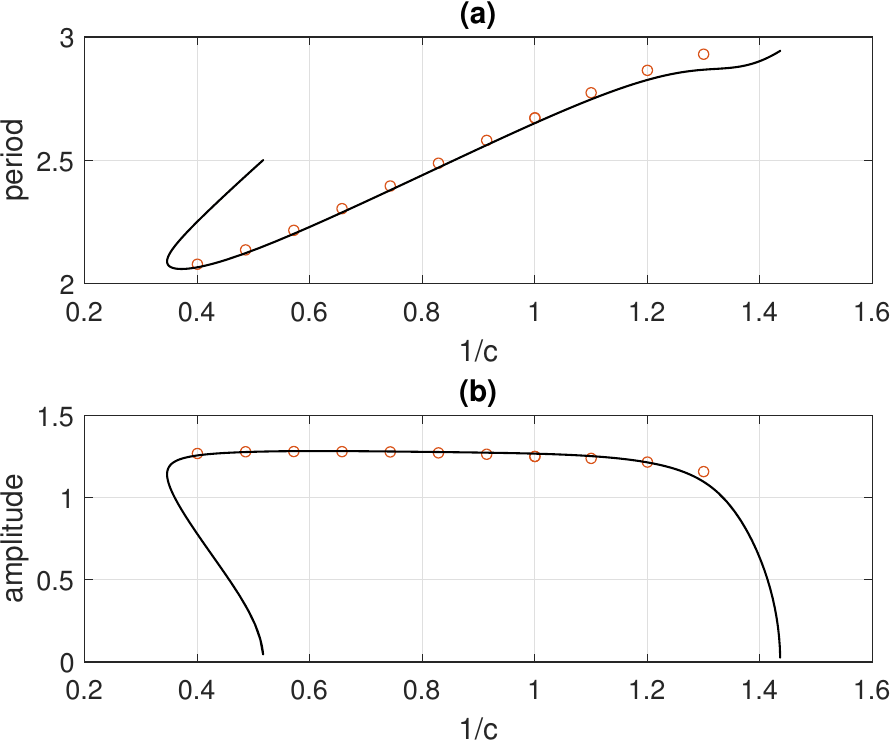}
\caption{Spatially-uniform periodic solutions
of~\eqref{Eq:MeanField}, \eqref{eq:IintC}.   
(a): Period $T$. (b): wave amplitude, defined as $\left(\sum_{m=1}^M \left[ a_m^2+b_m^2 \right] \right)^{1/2}$. Symbols are from numerical simulations 
of~\eqref{Eq:MeanField}, \eqref{eq:IintC}. This type of solution loses stability
to perturbations with spatial structure 
at a value of $1/c$ between $1.3$ and $1.35$, so is not stable upon its
creation at the rightmost Hopf bifurcation.
Parameters: $\tau = 1$, $\eta_0=0.1$, $M=5$, and other parameters as in Table~\ref{tab:param:2}.}
\label{fig:spatuni}
\end{figure}

\subsection{Periodic solutions with spatial structure}

Taking the values of the parameters
for which the neural field equation~\eqref{Eq:MeanField}
with the non-delayed input current~\eqref{eq:Iint0}
has a stable bump state, we performed a stability analysis of this state
for a fixed conduction speed $c = 1$ and a variable constant delay~$\tau$,
as explained in Sec.~\ref{sec:cond:bump}.
This allowed us to find several Hopf bifurcation points,
see Fig.~\ref{fig:percond}. By analogy with Sec.~\ref{sec:persln:inf},
we expect that these bifurcations give rise to periodic solutions
of Eq.~(\ref{Eq:MeanField}), (\ref{eq:IintC}).
Next, we describe the self-consistency argument
for calculating such periodic solutions.

Suppose that Eq.~(\ref{Eq:MeanField}), (\ref{eq:IintC}) has a solution $z = u(x,t)$
such that $u(x,t+T) = u(x,t)$ with some $T > 0$. 
%By the time rescaling $u(x,t) = z(x,t/\gamma)$
%where $\gamma = 2\pi / T$, we rewrite Eq.~(\ref{Eq:MeanField}) in the form
%\begin{equation}
%\pf{u}{t} = \fr{[ i ( \eta_0 + \kappa I(x,t) ) - \Delta ] ( 1 + u )^2}{2} - i \fr{(1 - u)^2}{2}
%\label{Eq:MeanField_}
%\end{equation}
%where
%$$
%I(x,t) = \int_{-\pi}^\pi G(x - y) 
%H\left( u\left(y,t - \tau-\fr{|x - y|}{c} \right) \right)  dy.
%$$
%So now instead of the $T$-periodic solution of Eq.~(\ref{Eq:MeanField})
%we will look for a $2\pi$-periodic solution of Eq.~(\ref{Eq:MeanField_}).
%
If $u(x,t)$ is even with respect to~$x$,
then the same is true for $W(x,t) = H( u(x,t) )$ and for
\[
I(x,t) = \int_{-\pi}^\pi G(x - y)
W\left(y,t - \tau-\fr{|x - y|}{c} \right) dy
= \int_{-\pi}^\pi G(\xi)
W\left(x+\xi,t - \tau-\fr{|\xi|}{c} \right) d\xi.
\]
(Recall that we consider an even coupling function $G(x)$.)
If we approximate $W(x,t)$ by its truncated Fourier series
$$
W(x,t) = \sum\limits_{n=0}^Q \left\{ w_{0n} + \sum\limits_{m=1}^M \left[ w_{mn}^\mathrm{c} \cos(m \gamma t) + w_{mn}^\mathrm{s} \sin(m \gamma t) \vphantom{\sum} \right] \right\} \cos(n x),
$$
where $\gamma=2\pi/T$, then
\begin{eqnarray*}
&&
W\left(x+\xi,t - \tau-\fr{|\xi|}{c} \right) = \sum\limits_{n=0}^Q \left\{ w_{0n} \vphantom{\sum\limits_{m=1}^M} \right. \\[2mm]
&&
+ \sum\limits_{m=1}^M \left[ w_{mn}^\mathrm{c} \cos( m \gamma(t-\tau)) \cos\left( \fr{m \gamma |\xi|}{c} \right) 
+ w_{mn}^\mathrm{c} \sin( m \gamma (t-\tau)) \sin\left( \fr{m \gamma |\xi|}{c} \right) \right. \\[2mm]
&&
\left.\left.+ w_{mn}^\mathrm{s} \sin( m \gamma(t-\tau)) \cos\left( \fr{m \gamma |\xi|}{c} \right) 
- w_{mn}^\mathrm{s} \cos( m \gamma(t-\tau)) \sin\left( \fr{m \gamma |\xi|}{c} \right) \right] \right\} \cos(n (x+\xi)).
\end{eqnarray*}
Furthermore,
\begin{eqnarray*}
I(x,t) &=& \sum\limits_{n=0}^Q \left\{ w_{0n} I_n(x)
+ \sum\limits_{m=1}^M \left[ w_{mn}^\mathrm{c} \cos( m \gamma(t-\tau)) I_{mn}^\mathrm{c}(x) 
+ w_{mn}^\mathrm{c} \sin( m \gamma(t-\tau)) I_{mn}^\mathrm{s}(x)
\vphantom{\sum\limits_{m=1}^M}  \right. \right. \nonumber \\[2mm]
&+&
\left.\left. \vphantom{\sum\limits_{m=1}^M}
w_{mn}^\mathrm{s} \sin( m \gamma(t-\tau)) I_{mn}^\mathrm{c}(x)
- w_{mn}^\mathrm{s} \cos( m \gamma (t-\tau)) I_{mn}^\mathrm{s}(x) \right] \right\},
\end{eqnarray*}
where
\begin{eqnarray*}
I_n(x) &=& \int_{-\pi}^{\pi} G(\xi)\cos{(n(x+\xi))} d\xi
= 2\cos{(nx)}\int_{0}^\pi G(\xi) \cos{(n\xi)}  d\xi,
%\nonumber
\\[2mm]
I_{mn}^\mathrm{c}(x) &=& \int_{-\pi}^\pi G(\xi)
\cos\left(\fr{m \gamma |\xi|}{c}\right) \cos(n (x+\xi))  d\xi
= 2\cos{(nx)}\int_{0}^\pi G(\xi)
\cos\left(\fr{m \gamma \xi}{c}\right) \cos{(n\xi)}  d\xi,
%\label{eq:Jcmn}
\\[2mm]
I_{mn}^\mathrm{s}(x) &=& \int_{-\pi}^\pi G(\xi) 
\sin\left(\fr{m \gamma |\xi|}{c}\right) \cos(n (x+\xi))  d\xi
= 2\cos{(nx)}\int_{0}^\pi G(\xi) 
\sin\left(\fr{m \gamma \xi}{c}\right) \cos{(n\xi)}  d\xi.
%\label{eq:Jsmn}
\end{eqnarray*}
Substituting~\eqref{Def:G} into the above formulas
and performing the integrals, we obtain
$$
I_n(x) =
\begin{cases}
2\pi A, & n=0, \\
B\pi\cos{x}, & n=1, \\
0, & n>1.
\end{cases}
$$
Moreover, for $I_{mn}^\mathrm{c}(x)$ we obtain an explicit formula
\begin{eqnarray*}
I_{mn}^\mathrm{c}(x) &=&
\cos{(nx)} \left\{ A \left[ \fr{c}{m \gamma - n c} \sin\left( \fr{m \gamma - n c}{c} \pi \right) + \fr{c}{m \gamma + n c} \sin\left( \fr{m \gamma + n c}{c} \pi \right)\right] \right. \\[2mm]
&+&
\frac{B}{2} \left[ \fr{c}{m \gamma - (n+1) c} \sin\left( \fr{m \gamma - (n+1) c}{c} \pi \right) + \fr{c}{m \gamma - (n-1) c} \sin\left( \fr{m \gamma - (n-1) c}{c} \pi \right) \right. \\[2mm]
&+&
\left. \left. \fr{c}{m \gamma + (n-1) c} \sin\left( \fr{m \gamma + (n-1) c}{c} \pi \right) + \fr{c}{m \gamma + (n+1) c} \sin\left( \fr{m \gamma + (n+1) c}{c} \pi \right) \right]\right\},
\end{eqnarray*}
and for $I_{mn}^\mathrm{s}(x)$ we obtain an explicit formula \small
\begin{eqnarray*}
I_{mn}^\mathrm{s}(x) &=&
\cos{(nx)} \left\{A \left[  \fr{c}{m \gamma - n c} \left( 1 - \cos\left( \fr{m \gamma - n c}{c} \pi \right) \right)  + \fr{c}{m \gamma + n c} \left( 1 - \cos\left( \fr{m \gamma + n c}{c} \pi \right) \right) \right] \right. \\[2mm]
&+&
\frac{B}{2} \left[ \fr{c}{m \gamma - (n+1) c}
\left( 1 - \cos\left( \fr{m \gamma - (n+1) c}{c} \pi \right) \right)
+ \fr{c}{m \gamma - (n-1) c} \left( 1 - \cos\left( \fr{m \gamma - (n-1) c}{c} \pi \right) \right) \right. \\[2mm]
&+&
\left. \left.
\fr{c}{m \gamma + (n-1) c} \left( 1 - \cos\left( \fr{m \gamma + (n-1) c}{c} \pi \right) \right)
+ \fr{c}{m \gamma + (n+1) c} \left( 1 - \cos\left( \fr{m \gamma + (n+1) c}{c} \pi \right) \right) \right] \right\}.
\end{eqnarray*}
\normalsize

For self-consistency calculations, the variables we use are the coefficients of the expansion of $W(x,t)$.
Given the set of the coefficients $w_{0n}$, $w_{mn}^\mathrm{c}$, $w_{mn}^\mathrm{s}$ with $n=0,1,\dots,Q$ and $m=0,1,\dots,M$, we construct $I(x,t)$.
Then, we calculate the $T$-periodic solution of Eq.~\eqref{Eq:MeanField}, $\tilde{u}(x,t)$,
using the technique explained in Sec.~\ref{sec:travinf}.
This solution is used to obtain a new function $\tilde{W}(x,t) = H( \tilde{u}(x,t) )$
and its Fourier coefficients $\tilde{w}_{0n}$,
$\tilde{w}_{mn}^\mathrm{c}$, $\tilde{w}_{mn}^\mathrm{s}$.
Finally, we equate all tilded coefficients to their un-tilded versions (and add a pinning
condition to find $T$).
%Interestingly, $\hat{w}_{0n}$ for $n>1$ are irrelevant, so we don't include them
%in the list of unknowns.

Following periodic solutions as $\tau$ is varied, we obtained several solution branches
for different values of $\eta_0$, see Fig.~\ref{fig:percond}.
The branches of solutions emanate from the stationary bump at the point it undergoes
Hopf bifurcations.
Numerical solutions of~\eqref{Eq:MeanField}, \eqref{eq:IintC} (circles) show that
the solution branches in Fig.~\ref{fig:percond} are stable between saddle-node bifurcations.

\begin{figure}[htbp]
\centering
\includegraphics[width=0.7\textwidth]{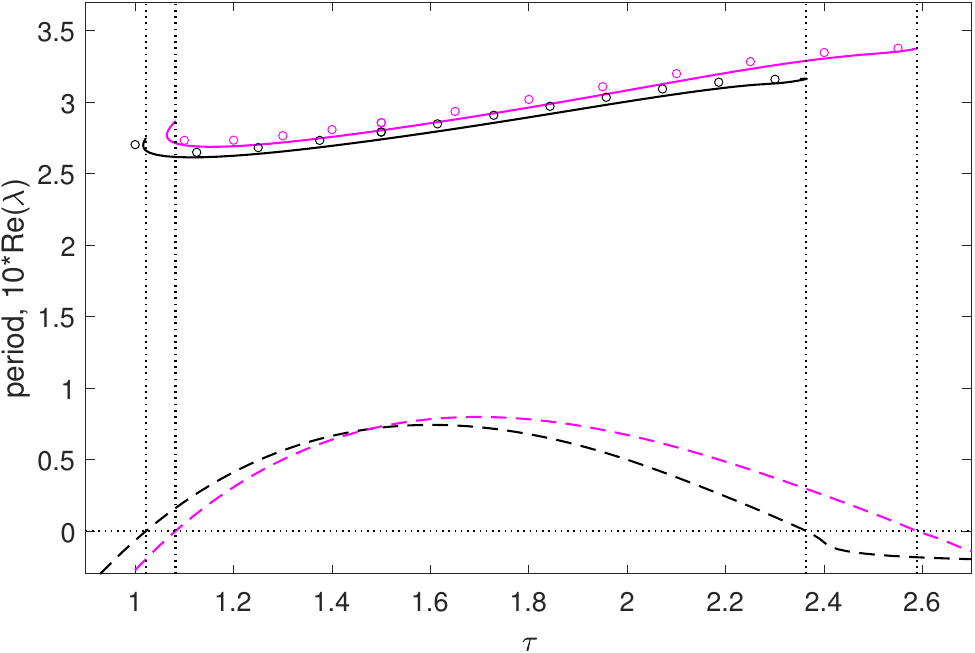}
\caption{Solid: period, $T$, of periodic solutions of~\eqref{Eq:MeanField}, \eqref{eq:IintC} 
with spatial structure, for 
$\eta_0=-0.15$ (black) and $\eta_0=-0.2$ (magenta).  Dashed: real part of the rightmost eigenvalues 
associated with the stability of a stationary bump solution, multiplied by 10
for the purpose of visualization. The dotted lines indicate where the eigenvalues cross
the imaginary axis. Symbols show the result of
numerically simulating~\eqref{Eq:MeanField}, \eqref{eq:IintC}.
Parameters: $c=1$, $M=5$, $Q = 20$, and other parameters as in Table~\ref{tab:param:2}.
}
\label{fig:percond}
\end{figure}

\section{Discussion}

We studied next generation neural field models derived from infinite
networks of heterogeneous synaptically coupled theta neurons on a ring. The interactions 
between neurons are delayed,
and we considered two forms of delays: distributed in time (two different examples), 
and conduction delays proportional
to the distance between neurons. Very little work has been done on such
models~\cite{byrode20,laiome23}. In each model we found the stability of spatially uniform
and bump states in terms of parameters describing the delays. When varying
delay parameters these states
undergo Hopf bifurcations leading to the creation of traveling waves and ``breathing'' bumps,
respectively. We showed how such time-dependent states can be described in a self-consistent 
way, and how to efficiently solve the corresponding self-consistency equations. 
We did not find Turing bifurcations leading to stationary patterns, for reasons explained
in Appendix~\ref{sec:noturing}.
Our results rely on choosing a small number of harmonics for the coupling function
$G(x)$, and the use of an efficient algorithm to find periodic solutions
of a periodically forced Riccati equation.

We now briefly discuss the relationships between our results and previous work.
Many authors have studied the effects of delays in classical neural field models,
often linearizing them around a spatially uniform 
state and determining the types of patterns that arise when it becomes 
unstable~\cite{atahut05,atahut06,vencoo07,vel13,roxmon11}. 
Typically, the patterns seen include stationary
Turing patterns, spatially uniform periodic states, and standing and traveling waves.
We also found spatially uniform periodic states and traveling waves that arise
from bifurcations of a spatially uniform state.
In addition, we found periodic solutions (breathing bumps) that arise from Hopf bifurcations
of stationary bumps, not often studied in systems with delays (but see~\cite{roxbru05}).

The recent paper~\cite{coomei25} considered a model similar to the one in
Sec.~\ref{sec:cond} but with discrete nodes and Wilson–Cowan dynamics at each node.
They constructed traveling waves using harmonic balance (writing periodic solutions as finite
Fourier series and imposing that they satisfy the governing differential equations at
a finite number of points in time) and obtained results similar to those in
Fig.~\ref{fig:trav}. One advantage of their approach is that they can determine the
stability of traveling waves.
We now discuss a number of options for future work.

%\cite{roxbru05} considered a fixed delay with a piecewise linear firing rate function.
%The stability of stationary bump, oscillating uniform and traveling wave states could be found analytically. Generalized in~\cite{roxmon11} to arbitrary firing rate function and connectivity.

%\cite{atahut06} considered a distribution of both delays and of conduction speeds and looked
%at stability of spatially uniform state. Found 
%spatially uniform periodic states, traveling waves and Turing patterns. Also studied traveling fronts.

%\cite{vencoo07} considered conduction delays. Looked at stability of spatially uniform state.
%Found standing and traveling waves. Also spatially uniform periodic states and Turing patterns.

%\cite{fayfau10} consider a general model with general delay and discuss existence, uniqueness, 
%and asymptotic stability of a solutions. They also develop a numerical scheme for the case
%of transmission delays and show the results of simulations.

%\cite{vel13} considered a model with delays like in Sec.~\ref{sec:cond} and linearized around
%the spatially uniform steady state. Found various bifurcations to traveling waves, standing waves.
%Found more exotic solutions.
%

While we have considered networks of theta neurons, they are equivalent to quadratic
integrate-and-fire (QIF) neurons, so our methods could be used to study networks of QIF neurons,
such as~\cite{omelai24}, but with delays~\cite{ratpyr18}. 
Winfree oscillators~\cite{aristr01} are another type of
model neuron and the Ott/Antonsen ansatz can be used to study infinite networks of 
them~\cite{pazmon13}. However, there seems to be little, if any, research on delayed
Winfree oscillators.

We considered several types of delays: distributed and conduction.
Other types of delays that could be considered include activity-dependent
delays~\cite{noopar20,parlef20}, time-dependent delays, and state-dependent delays~\cite{harkri06}.
One point of interest is that we assumed that the conduction velocity has a single value; 
instead, there may be a
distribution of velocities~\cite{bojlil10,atahut06}. This translates to a distribution of
conduction delays, even for a fixed distance between neurons. It may be possible to extend
our results to include a distribution of conduction velocities. While the dynamic solutions
we studied arose from the presence of delays, periodic solutions amenable to our
approach will also occur in periodically forced systems~\cite{reyhug22}.

We studied a one-dimensional domain with periodic boundary conditions, but a more
realistic domain may be the surface of a sphere~\cite{visnic17,daicec20,spevan24}.
The methods presented in this paper can be applied in this case as well.
See for example~\cite{batcle23}, where they were used
to study undelayed systems of phase oscillators in a two-dimensional domain.

\appendix

\section{Proof of no Turing bifurcations}
\label{sec:noturing}
In this section we prove that there are no Turing bifurcations of spatially uniform states
creating stationary patterns. We do this by showing that there are no zero eigenvalues
associated with the linearization about these uniform states. 

\subsection{Distributed delays}
For the models in Secs.~\ref{sec:inf} and~\ref{sec:compact} the relevant eigenvalues associated with 
the spatially uniform state $Z_3$ are given by~\eqref{eq:eigB}. For zero to be an eigenvalue we need
\be
   |I_2-\pi BL(0)|=0, \label{eq:stabB}
\ee
where
\[
   L(0)=i\kappa\bp D'(Z_3) & 0 \\ 0 & \overline{D'(Z_3)} \ep \bp -\mu & 0 \\ 0 & -\overline{\mu} \ep^{-1}
\bp \zeta & \zeta \\ - \overline{\zeta} & -\overline{\zeta} \ep
\]
and where $\mu,\zeta$ and $D'(Z)$ are given by~\eqref{Def:mu},~\eqref{Def:zeta} and~\eqref{Def:D_}, respectively.
Now~\eqref{eq:stabB} is equivalent to $1+2\pi\kappa B\Real(i \zeta D'( Z )/\mu)=0$. This is not true
for the parameters in Table~\ref{tab:param}, so~\eqref{eq:stabB} does not hold.

\subsection{Conduction delays}
For the model in Sec.~\ref{sec:cond} the eigenvalues associated with 
a spatially uniform state are given by~\eqref{Eq:K_k}. For zero to be an eigenvalue
we need
\be
   |K_k(0)|=0. \label{eq:stabK}
\ee
For $k>1$, $Q_k(0)=0$, so~\eqref{eq:stabK} is equivalent to $|\mu|^2=0$, where $\mu$
is given by~\eqref{Def:mu}. Generically, $\mu\neq 0$, and for the parameters used
in Fig.~\ref{fig:conduni}, $\mu\neq 0$. Thus~\eqref{eq:stabK} cannot hold for $k>1$.

When $k=1$, $Q_1(0)=\pi B$ and~\eqref{eq:stabK} is equivalent to
\[
   \left|\bp \mu+i \pi \kappa B \zeta D'( Z ) & i \pi \kappa B \zeta \overline{D'( Z )} \\ -i \pi \kappa B \overline{\zeta} D'( Z ) & \overline{\mu}-i \pi \kappa B \overline{\zeta} \overline{D'( Z )} \ep\right|=0,
\]
where $\zeta$ and $D'(Z)$ are given by~\eqref{Def:zeta} and~\eqref{Def:D_}, respectively, 
or $|\mu|^2+2\pi\kappa B\Real(i\bar{\mu} \zeta D'( Z ))=0$. Again, this is not true for
the parameters used in Fig.~\ref{fig:conduni}, so~\eqref{eq:stabK} does not hold for $k=1$
either. (The case $k=0$ would not lead to a spatial pattern.)

\section{Solutions in discrete network}
\label{sec:disc}
In this section we show examples of the solutions found in the continuum limit, but in
discrete networks.

\subsection{Distributed delays --- infinite support}
For the model in Sec.~\ref{sec:inf}, performing the same manipulations as in that section, we have
\be
   \frac{d\theta_j}{dt}=1-\cos{\theta_j}+(1+\cos{\theta_j})(\eta_j+\kappa I_j) \label{eq:dthetadtA}
\ee
where
\be
   \frac{dI_j}{dt}=(Y_j-I_j)/a
\ee
and
\be
   \frac{dY_j}{dt}=\frac{1}{a}\left[\frac{2\pi}{N}\sum_{k=1}^N G_{jk}P(\theta_k(s-\tau))-Y_j\right]. \label{eq:dYdtB}
\ee
A traveling wave solution is shown in Fig.~\ref{fig:trinf}.
An approximately periodic solution is shown in Fig.~\ref{fig:perinf}.
Spatially uniform and stationary bump states were also found (not shown).

\begin{figure}[htbp]
\centering
\includegraphics[width=0.9\textwidth]{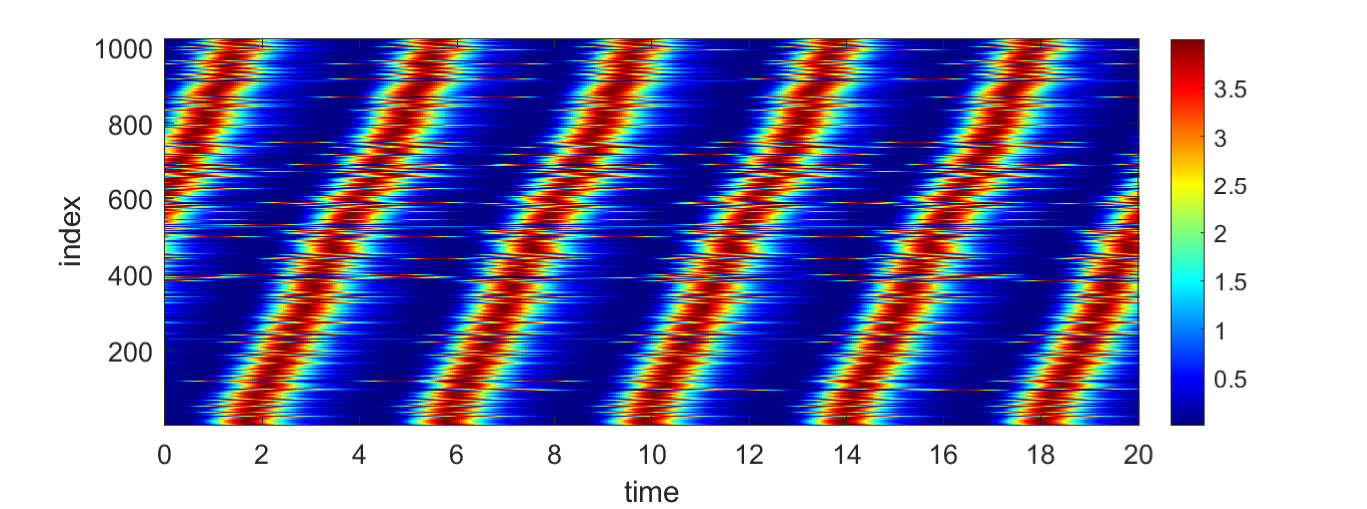}
\caption{A traveling wave solution of~\eqref{eq:dthetadtA}--\eqref{eq:dYdtB}.
Color shows $(1-\cos{\theta_j})^2$.
Parameters as in Table~\ref{tab:param}, with
$N=1024,\tau=2,a=1$.}
\label{fig:trinf}
\end{figure}

\begin{figure}[htbp]
\centering
\includegraphics[width=0.9\textwidth]{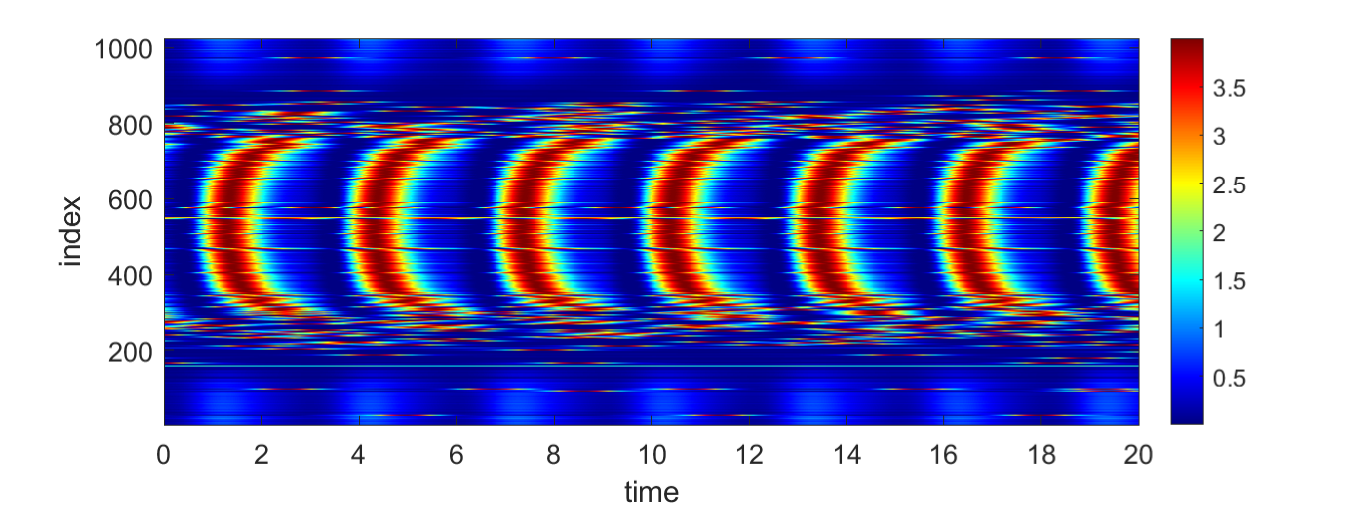}
\caption{An approximately 
periodic solution of~\eqref{eq:dthetadtA}--\eqref{eq:dYdtB} arising from a Hopf
bifurcation of a stationary bump. Color shows $(1-\cos{\theta_j})^2$.
Parameters as in Table~\ref{tab:param}, with
$N=1024,\tau=2,a=0.2$.}
\label{fig:perinf}
\end{figure}

\subsection{Distributed delays --- compact support}
For the model in Sec.~\ref{sec:compact} we have~\eqref{eq:dthetadtA} where
\be
   I_j(t)=\frac{2\pi}{aN}\sum_{k=1}^N G_{jk}\int_{\tau}^{\tau+a} P(\theta_k(t-s))\: ds. \label{eq:IrectA}
\ee
We approximate the integral using the $n$ point Gauss-Legendre quadrature rule~\cite{tre00}:
\[
   \frac{1}{a}\int_{\tau}^{\tau+a} P(\theta_k(t-s))\: ds \approx \sum_{i=1}^n w_iP(\theta_k(t-[ax_i/2+a/2+\tau]))
\]
where $x_i$ is the $i$th root of the $n$th Legendre polynomial, $P_n(x)$, and the weights 
(normalized to sum to 1) are
\[
   w_i=\frac{1}{(1-x_i^2)P'_n(x_i)}.
\]
We set $n=5$.
A traveling wave for $\tau=2,a=3$, and other parameters as in Table~\ref{tab:param}, looks
almost indistinguishable from the one shown in Fig.~\ref{fig:trinf}.
A periodic solution for $\tau=2,a=1$, and other parameters as in Table~\ref{tab:param}, looks
almost indistinguishable from the one shown in Fig.~\ref{fig:perinf}.
A solution like that in Fig.~\ref{fig:other}(a) was found, as were spatially uniform
and steady bump solutions (not shown).

\subsection{Conduction delays}

We now show some solutions of~\eqref{eq:dthetadt}, \eqref{eq:IdiscB}. Traveling waves with different angular wave numbers
are shown in Fig.~\ref{fig:twist} and an approximately periodic solution is shown in 
Fig.~\ref{fig:brebump}. Spatially uniform stationary
and periodic solutions were also found (not shown). 

\begin{figure}[htbp]
\centering
\includegraphics[width=0.8\textwidth]{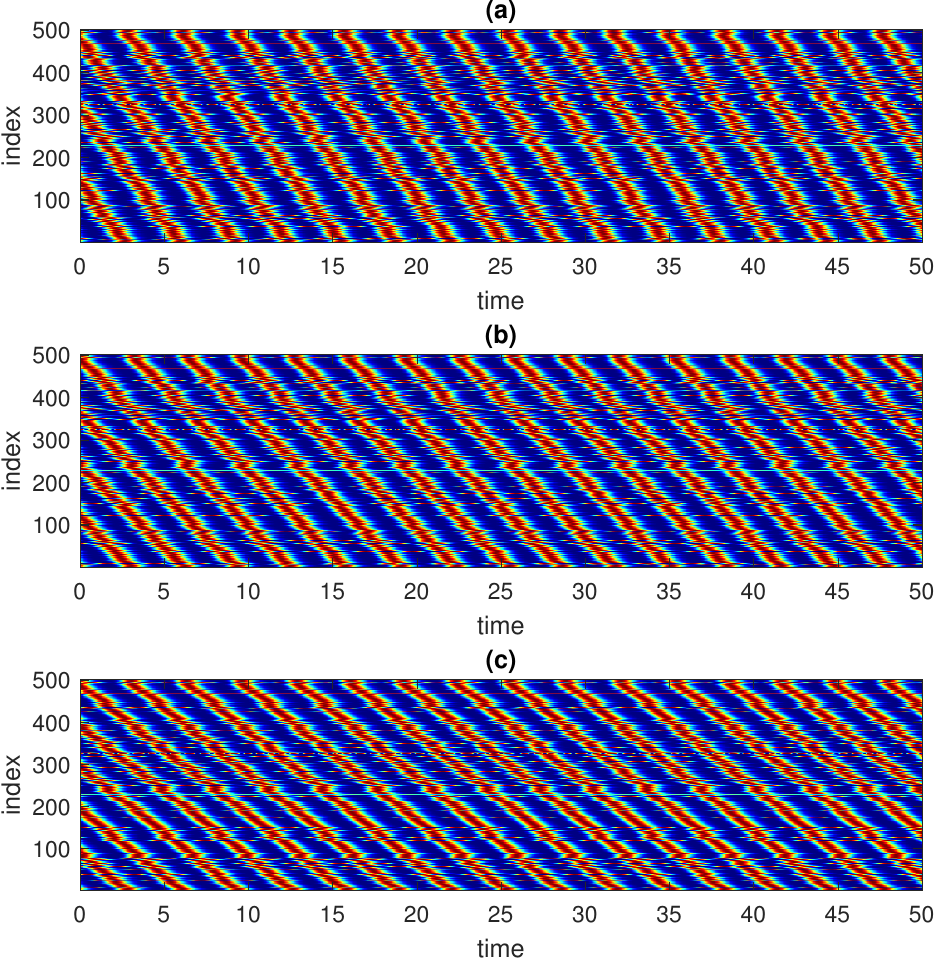}
\caption{Traveling waves with different angular wave numbers
for~\eqref{eq:dthetadt}, \eqref{eq:IdiscB}. (a): $k=3, 1/c=0.5$. (b): $k=4, 1/c=1$. (c): $k=5, 1/c=1/5$.  Parameters as in Table~\ref{tab:param:2} with
$\eta_0=0.1,\tau=0.2,N=501$. The same $\eta_i$
were used for all simulations. Color map is as in Fig.~\ref{fig:perinf}.}
\label{fig:twist}
\end{figure}

\begin{figure}[htbp]
\centering
\includegraphics[width=0.8\textwidth]{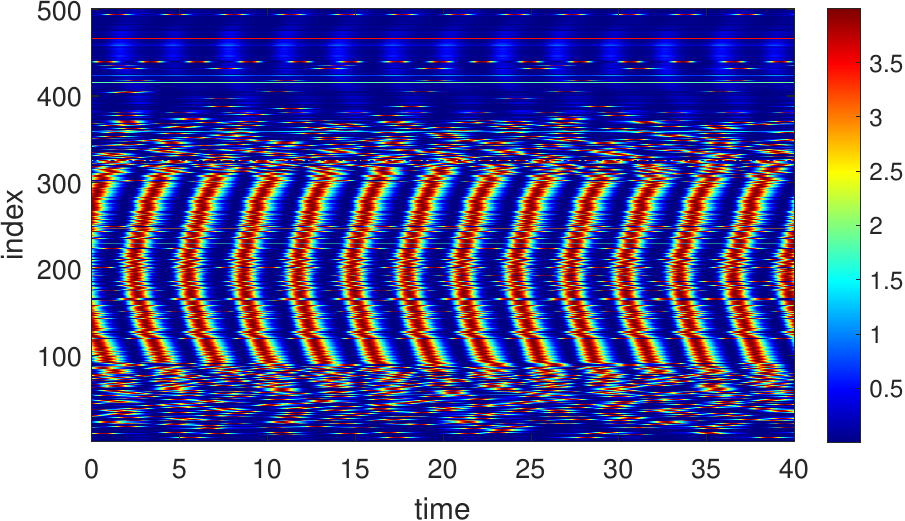}
\caption{An approximately periodic solution
for~\eqref{eq:dthetadt}, \eqref{eq:IdiscB}.   Parameters as in Table~\ref{tab:param:2} with
$\eta_0=-0.2,c=1,\tau=2,N=501$.  $(1-\cos{\theta_j})^2$ is shown in colour. }
\label{fig:brebump}
\end{figure}

%\section*{Acknowledgments}
%The work of O.E.O. was supported by the Deutsche Forschungsgemeinschaft under Grant No. OM 99/2-3.

%\bibliographystyle{siamplain}
%\bibliography{spiral}

\end{document}